\documentclass{article}

\UseRawInputEncoding
\usepackage[all]{xy}
\usepackage{amssymb}
\usepackage{comment}
\usepackage{stmaryrd}
\usepackage{amsmath,amsthm} 
\usepackage{graphicx}
\usepackage{listings}
\usepackage{enumerate}
\usepackage{url}
\urldef\myurl\url{foo%.com}
\usepackage{hyperref}

\newtheorem{definition}{Definition}[section]
\newtheorem{theorem}{Theorem}[section]
\newtheorem{lemma}{Lemma}[section]
\newtheorem{conjecture}{Conjecture}[section]
\newtheorem{proposition}{Proposition}[section]
\newtheorem{informal-proposition}{Informal proposition}[section]
\newtheorem{problem}{Problem}[section]

\newtheorem{unknown-proposition}{unknown-Proposition}[section]

\newtheorem{corollary}[theorem]{Corollary}

\newtheorem{example}{Example}[section]

\newcommand{\rat}{\mathbb{Q}}
\newcommand{\Int}{\mathbb{Z}}

\begin{document}

\title{Rings with common division, common meadows and their conditional equational theories\\
\author{Jan A Bergstra \\
Informatics Institute, University of Amsterdam,\\
 Science Park 904, 1098 XH, Amsterdam, 
The Netherlands\\
j.a.bergstra@uva.nl\\ \and \\
John V Tucker\\
Department of Computer Science,\\ Swansea University, Bay Campus, \\Fabian Way, Swansea, SA1 8EN, United Kingdom\\j.v.tucker@swansea.ac.uk }}

\maketitle

\begin{abstract}
\noindent We examine the consequences of having a total division operation $\frac{x}{y}$ on commutative rings.  We consider two forms of binary division, one derived from a  unary inverse, the other defined directly as a general operation; each are made total by setting $1/0$ equal to an error value $\bot$, which is added to the ring. Such totalised divisions we call \textit{common divisions}. In a field the two forms are equivalent and we have a finite equational axiomatisation $E$ that is complete for the equational theory of fields equipped with common division, called \textit{common meadows}. These equational axioms $E$ turn out to be true of commutative rings with common division but only when defined via inverses. We explore these axioms $E$ and their role in seeking a completeness theorem for the conditional equational theory of common meadows. We prove they are complete for the conditional equational theory of commutative rings with inverse based common division. By adding a new proof rule, we can prove a completeness theorem for the conditional equational theory of common meadows.
Although, the equational axioms $E$ fail with common division defined directly, we observe that the direct division does satisfies the equations in $E$ under a new congruence for partial terms called eager equality. 
\end{abstract}

\paragraph{Keywords \& phrases:}
commutative rings, division operators, meadows, common division, common meadows, equations, conditional equations, varieties, quasivarieties


\section{Introduction}

Arithmetical structures are many and varied. Classically, we think of the integers $\mathbb{Z}$, rationals $\mathbb{Q}$, reals $\mathbb{R}$, and complex numbers $\mathbb{C}$, as well as modulo arithmetics $\mathbb{Z}_n$, $p$-adics and derived systems. To these should be added a number of arithmetic structures developed for computer computation such as fixed point, floating point and interval arithmetics. Whilst the algebra of the classical structures is very well understood through the axiomatic theories of rings and fields, the algebra of the computer arithmetics is barely developed. This is a problem as the algebraic properties of these structures are fundamental in specifying, designing and reasoning about computer programs, e.g., in predicting precision, anticipating and finding errors, and subsequent validation and verification. 

In a series of papers, we have been investigating a few key semantic features of computer arithmetics and exploring their algebraic properties. Our approach is to start from field theory, especially the field of rational numbers, and impose algebraic conditions on fields that represent some of the semantic constraints placed on computer arithmetics, such as bounds, finiteness, and the need for additional operations. To develop a theory of abstract arithmetic data types for computer programming, we apply methods that have been developed from universal algebra and logic.

So, the axiomatic concept of a field has a number of shortcomings as a specification of an arithmetical structure for computing. First, a field is a commutative ring with 1 in which all non-zero elements are invertible. Thus, strictly speaking, division is not an operator in a field.   Thus, we add division $x \div y$ to the primary field operations of $x + y, - x, x \cdot y$ to make an algebra we call a \textit{meadow} \cite{BergstraT2007}. Secondly, in computing, applying an operator \textit{must} return a value for \textit{all} arguments, i.e., an operator must be a total function. Thus, for our meadows we must find a way of defining a value for $x \div 0$.


\subsection{Totalising division}

We have studied several ways of making division total in a meadow -- see Section \ref{concluding_remarks} for some background. One way stands out: for all $x$, 
$$ \ x \div 0 = \bot$$
where $\bot$ is a new element added to the meadow and having the property of \textit{absorption}: if $\bot$ is an argument to any operation then it returns the value $\bot$. We think of $\bot$ as an error value.  Such an algebra we call a \textit{common meadow}, which has the form of an enlargement of a field $F$:
$$(F \cup \{ \bot \} \ | \ 0, 1, \bot, x+y, -x, x \cdot y, x \div y).$$ 
A common meadow of rational numbers can serve as an idealised model of the simple arithmetic of the pocket calculator. Let $\mathsf{CM}$ be the class of common meadows.

Common meadows were introduced in \cite{BergstraP2015}. The introduction of $\bot$ into a field invalidates some key algebraic laws such as
\begin{center}
$x + (- x) = 0$ and $x \cdot 0 = 0$ because $x + (- \bot) = x + \bot = \bot$ and $x \cdot \bot  = \bot$.
\end{center}
Subsequently, we have found sets of equations to axiomatise calculations in common meadows, among which is a set $E_{\mathsf{ftc-cm}}$ studied in \cite{BergstraT2022CJ,BergstraT2023arxiv}; the subscript $\mathsf{ftc-cm}$ stands for \textit{fracterm calculus for common meadows} and the equations are listed in Table \ref{EwcrBot} and Table \ref{FTCcm}  in Section \ref{the_equational_axiom_set} below. These equations have proved to be a useful tool for understanding the effects of totalising division, intuitively and formally. The axioms are close to classical axioms and entail an important property for working with fractions, called \textit{flattening}.  From universal algebra and logic we know something about the class of models of $E_{\mathsf{ftc-cm}}$.  

\subsection{Axiomatising common division}
In this paper, we focus on the equational axioms in  $E_{\mathsf{ftc-cm}}$ and, in particular, on the scope of their semantics which extend beyond the common meadows $\mathsf{CM}$. First, we examine the consequences of having a total division operation $x \div y$ on commutative rings rather than fields.  We consider two forms of division: the first is derived indirectly from a unary inverse operator, which we will refer to as \textit{inverse based common division};  later, the second is defined for integral domains directly as a binary operation, which we will refer to as \textit{direct common division}.  Both inverse based common division and directly defined common division are made total by setting $x/0$ equal to the error value $\bot$.  In a field the two forms of division are equivalent and so we have the option of using either; this we have done in developing some common meadow theory for which the axiomatistion $E_{\mathsf{ftc-cm}}$ was designed.

Now, as pointed out in \cite{DiasD2023}, in addition to common meadows, the class $\mathsf{CR_{\div, \bot}}$ of commutative rings with $1$ equipped with an inverse based common division operator \textit{also} satisfy $E_{\mathsf{ftc-cm}}$. We axiomatise the class of common meadows relative to the commutative rings with inverse common division using this conditional equation, \textit{Additional Value Law},
$$\mathsf{AVL}: \ \  \frac{1}{x} = \bot \to 0 \cdot x = x.$$

From \cite{BergstraT2023arxiv}, we know that $E_{\mathsf{ftc-cm}}$ is a finite equational base for the equational theory of common meadows from which it follows that $E_{\mathsf{ftc-cm}}$ is a finite equational base for the equational theory of commutative rings with common division.  Here we prove the theorem (\ref{completeness_for_rings}):

\begin{theorem}
$E_{\mathsf{ftc-cm}}$ provides a complete equational axiomatisation of the conditional equational theory of the class of commutative rings with inverse based common division.
\end{theorem}

From  \cite{BergstraT2022CJ}, we know that $E_{\mathsf{ftc-cm}}$ is  \textit{not} a finite base for the conditional equational theory of common meadows. 
Here we show that $E_{\mathsf{ftc-cm}} + \mathsf{AVL}$ does not axiomatise the conditional equational theory of common meadows. However, on adding a simple proof rule  $\mathsf{R_{cm}}$ to conditional equation logic to define proof system $\vdash_{R_{cm}}$ we can prove completeness:

\begin{theorem}  
$\mathsf{CM} \models e_1 \wedge e_2 \ldots \wedge e_n   \to e$ if, and only if, 
$E_{\mathsf{ftc-cm}}+ \mathsf{AVL} \vdash_{\mathsf{R_{cm}}} E  \to e_1 \wedge e_2 \ldots \wedge e_n   \to e$.
\end{theorem}

This leaves us with the open question:

\begin{problem} 
\label{MainProblem}
Is the conditional equational theory of the class $\mathsf{CM}$ of common meadows axiomatisable with finitely many conditional equations?
\end{problem}

We conjecture the answer: no.

Finally, in contrast to the inverse defined division,  we note that the directly defined common division \textit{does not} satisfy all the axioms of $E_{\mathsf{ftc-cm}}$ over all rings. Surprisingly, this latter common division \textit{does} satisfy the equations in $E_{\mathsf{ftc-cm}}$ under a new equality relation for partial algebras called \textit{eager equality} \cite{BergstraT2022_ToCL,BergstraT2023_CJ}. 


\subsection{Structure of the paper}

In Section \ref{algebras_and_equations}, we describe the totalisation of partial operations by the addition of $\bot$ in the general case of an arbitrary algebraic structure. In Section \ref{division}, we summarise the situation for rings and fields, introducing the division operators based on inverse, and the common meadow axiomatisation $E_{\mathsf{ftc-cm}}$ and its model class. In Section \ref{rings_with_bot}, we prove a number of basic facts about commutative rings with inverse based common division, including results about their relationship with common meadows. In Section \ref{homomorphisms} we examine some homomorphisms. After  some examples in Section \ref{conditional_equations}, in the  Section \ref{customising_logic}  we customise the logic to prove the completeness theorem for common meadows.  In Section \ref{integral_domains}, we look at integral domains and at this point, we introduce the second form of division. In Section \ref{Rings_with _division _and_eager_equality}, we introduce eager equality and observe the validity of $E_{\mathsf{ftc-cm}}$ for the second form of division. Lastly, in Section \ref{concluding_remarks}, we sketch our investigation into the various ways of totalising division, which is the background to this paper.

We need only the basics of ring and field theory \cite{vanderWaerden1970}, universal algebra \cite{Mal'tsev1973,BurrisSankapannavar1981,MeinkeTucker1992}, and abstract data type theory \cite{EhrichWL1997}.


\section{Algebras and equations with an absorptive element}\label{algebras_and_equations}

The idea of making partial operations total with an absorptive element is quite general. We will summarise the process for an algebraic structure as discussed in \cite{BergstraT2022TM}.\footnote{All algebras are assumed to be single-sorted in this paper.} The basic algebra we need is to be found in the classic works of Birkhoff \cite{Birkhoff1935} and, later, Mal'cev \cite{Mal'tsev1973}.


\subsection{Algebras with absorption}\label{algebras_with_absorption} 

Let $A$ be an algebra with signature $\Sigma$, which may have partial operations.  The definition of an absorptive element for  a general algebraic structure is as follows: 

\begin{definition}\label{adding_bottom}  
A $\Sigma$-algebra $A$ has an {\em absorptive element} $a \in A$ if for each operation $f$ of $\Sigma$, if one of its arguments equals $a$ then so does its value.
\end{definition}

The following partial result is relevant for our arithmetical data types as they are equipped with at least one binary operation (Proposition 2.1 of~\cite{BergstraT2022TM}).

\begin{lemma}\label{uniqueness} 
If there is at least one function $f$ in $\Sigma$ with two or more arguments on  $\Sigma$-algebra $A$ then an absorptive element in  the algebra $A$  is unique.
Let $\bot$ be a standard notation for an absorptive element.
\end{lemma}

Our use of $\bot$ in an algebra will imply that it is absorptive with respect to all the operations of the algebra.

\begin{definition} 
Let $A$ be a $\Sigma$-algebra and assume $\bot \not \in A$. Then the $\bot$-{\em enlargement} of $A$ is an algebra $A_\bot$ with signature $\Sigma_\bot = \Sigma \cup \{\bot \}$ and carrier set  $A_\bot = A \cup \{\bot\}$, which is an enlargement of $A$ in which  $\bot$ is absorptive.
\end{definition}

\begin{definition}
As a notation we will use $A_\bot$ for the construction of the $\bot$-extension of $A$. For any class $\mathsf{K}$ of $\Sigma$-algebras, define 
$$\mathsf{K}_\bot = \{ A_\bot \  |  \ A \in \mathsf{K}\}.$$
\end{definition}

Note we have only added $\bot$ to an algebra; we have not yet employed it to make partial operations total.


\subsection{Equations and absorption}\label{equations_and_absorption} 

First, we consider the definability of the classes. Let $\Sigma$ be a signature and $T(\Sigma, X)$ be the algebra of all terms over $\Sigma$ containing variables from $X$. The value of a term $t \in T(\Sigma, X)$ on elements $a$ from a $\Sigma$ algebra $A$ is denoted $\llbracket t \rrbracket(a)$ or $t(a)$. If $X = \emptyset$ we write $T(\Sigma)$ for the algebra  of all closed terms.

Let $Eqn(\Sigma)$ and $CEqn(\Sigma)$ be the set of all equations $t = r$ and conditional equations $t_1 = r_1 \wedge \ldots \wedge t_n = r_n \to t=r$ over $\Sigma$, respectively. Let $\mathsf{Alg}(\Sigma, E)$ be the class of all $\Sigma$ algebras satisfying the axioms of $E \subset Eqn(\Sigma)$ or of $E \subset CEqn(\Sigma)$. 
Common terms for these types of axiomatisable classes are \textit{variety} and \textit{quasivariety}, based on equivalent characterisations of these classes in terms of algebraic constructions \cite{MeinkeTucker1992}. 

The class $\mathsf{Alg}(\Sigma, E)$ has an initial algebra $I(\Sigma, E)$ that can be constructed as a quotient  $T(\Sigma)/\equiv$ by:
$$T(\Sigma)/\equiv \  \models t = t'   \iff E \vdash  t = t',$$
wherein $\vdash$ is proof via equational logic or conditional equational logic. By equational logic we mean manipulating equations on the basis of reflexivity, symmetry, transitivity of equality, and the substitution of equal terms into equations. The less well known conditional equational logic is more involved requring rules to access the conditions \cite{MeinkeTucker1992}. However, from Birkhoff and Mal'cev, we can also rely on first order derivations: 

\begin{lemma}
Equations and conditional equations are provable from $E$ by equational and conditional equational logic respectively if, and only if, they are provable from $E$ in first order logic. 
\end{lemma}

The following is Proposition 2.2  in \cite{BergstraT2022TM}:

\begin{theorem}\label{NonQ-Variety}
Let $\Sigma$ be a signature which contains at least one function symbol $f$ with two or more arguments. For any nonempty 
class $\mathsf{K}$ of $\Sigma$-algebras in which $f$ is total, the class $\mathsf{K}_\bot$ cannot be defined by conditional equations, i.e., is not a quasivariety.
\end{theorem}

If only unary functions are present in $\Sigma$ then $\mathsf{K}_\bot$ can be specified by the conditional equations $f(x) = \bot \to x = \bot$ for each operation $f$. 

It follows from Theorem~\ref{NonQ-Variety} that in most cases there is no way to  find a  conditional equational specification for $\mathsf{K}_\bot$ that is complete for first order formulae. In particular:  

\begin{corollary} 
\label{NonQVar}
Given a signature $\Sigma$ which contains at least one function symbol $f$ with two or more arguments, and any equationally defined class $\mathsf{K} = \mathsf{Alg}(\Sigma, E)$ in which $f$ is total, then the class $\mathsf{K}_\bot$ is {\em not} definable by conditional equations, i.e., is not a quasivariety.
\end{corollary}

Rather than looking for equational axioms to capture the algebras of $\mathsf{K}_\bot$, what can be done instead, is to look for a finite axiomatisation of the full equational, or conditional equational, theory of $\mathsf{K}_\bot$. In general:

\begin{definition}
The {\em equational theory of the class $\mathsf{K}$} is the set 
$$EqnThy(\mathsf{K}) =  \{ e \in Eqn(\Sigma) \ | \  \forall A \in  \mathsf{K}. A \models e \} $$ 
of all equations  over $\Sigma$ that are true in all the algebras in $\mathsf{K}$. The {\em conditional equational theory of the class $\mathsf{K}$} is the set 
$$CEqnThy(\mathsf{\mathsf{K}}) =  \{ c  \in CEqn(\Sigma) \ | \  \forall A \in  \mathsf{K}. A \models c \} $$ 
of all conditional equations over $\Sigma$ that are true in all the algebras in $\mathsf{K}$.
\end{definition}

\begin{lemma}\label{subclass_lemma}
Let $\mathsf{K}$ and $L$ be classes of $\Sigma$ algebras.
If $\mathsf{K} \subset L$ then $EqnThy(\mathsf{L}) \subset EqnThy(\mathsf{K})$ and $CondEqnThy(\mathsf{L}) \subset CondEqnThy(\mathsf{K})$.
\end{lemma}

\begin{definition}\label{def:completeness}
A set $E$ of equations or conditional equations is {\em complete} or is a {\em basis} for the equational theory, or conditional theory, of the class $\mathsf{K}$ if, respectively,
$$e \in EqnThy(\mathsf{K}) \iff E \vdash  e \   \textrm{or} \ e \in CEqnThy(\mathsf{K}) \iff E \vdash  e.$$
\end{definition}

%
%


\subsection{$\bot$-enlargement for partial algebras}
\label{PAs}

Notice that earlier $A$ was total or partial and $\mathsf{K}_\bot$ simply added $\bot$ and ensured the operations respected it. The notion of a $\bot$-extension applied to total algebras does little useful work. However, it is useful to partial algebras:  suppose the application of a function $f$ on arguments is undefined in a partial $\Sigma$-algebra $A$ then we can take $\bot$ as the result of $f$ on those arguments in $A_\bot$:

\begin{definition}\label{totalising_partial_operations}
Let $A$ be a partial algebra containing at least two elements.
Then define the enlargement $B=\mathsf{Enl}_\bot(A)$ to be a total algebra as follows: $B = A \cup \{\bot\}$ where $\bot$ is an absorptive element, and for $b_1, ..., b_n \in A$
\begin{center}
$f(b_1,\dots,b_n) = \bot$ in $B$ if, and only if, $f(b_1,\dots,b_n)$ is undefined in $A$.
\end{center}
\end{definition}

Following Definition \ref{totalising_partial_operations}, in the case of $A$ partial, the enlargement notation $\mathsf{Enl}_\bot(A)$ implies that the partial operations of $A$ have been made total using $\bot$. Of course, $A_\bot$ is a partial algebra.

\begin{definition}
 For any class $\mathsf{K}$ of $\Sigma$-algebras, define 
$$\mathsf{Enl}_\bot(\mathsf{K}) = \{ \mathsf{Enl}_\bot(A) \  |  \ A \in \mathsf{K}\}.$$
\end{definition}

Conversely, we can go the other way from a total algebra with absorptive element $\bot$ to a partial algebraic data type:

\begin{definition} 
Let $A$ be an algebra with $\bot$ as an absorptive element, and such that $A$ has at least two elements.
Then define the transformation $B=\mathsf{Pdt}_\bot(A)$ to be a partial algebra as follows: $B = A -\{\bot\}$ and for $b_1, ..., b_n, b \in B$,
\begin{center}
$f(b_1,\dots,b_n) = b$ in $B$ if, and only if, $f(b_1,\dots,b_n) = b$ in $A$.
\end{center}
\end{definition}

\begin{proposition}\label{Enl_Pdt}
Let $A$ be a total algebra with $\bot$ as an absorptive element, 
and such that $A$ has at least two elements.  Then 
\begin{center}
$\mathsf{Enl}_\bot(\mathsf{Pdt}_\bot(A)) = A$.
\end{center} 
\end{proposition}

\begin{proposition}\label{Pdt_Enl}
Let $A$ be an algebra. Then 
\begin{center}
$\mathsf{Pdt}_\bot(\mathsf{Enl}_\bot(A)) = A$.
\end{center} 
\end{proposition}


\section{On axiomatising division}\label{division}


\subsection{Rings, fields and meadows}\label{division_inverse}

A commutative ring with $1$ is a field if every non-zero element is invertible. In particular, rings and fields have the same operations to which we will add division:

\begin{lemma}\label{dividers_are_well-defined}
Let $R$ be a commutative ring with 1.  If $a \cdot b =1$ and $a \cdot c = 1$ then $b=c$.
\end{lemma}
\begin{proof}
Assume $a \cdot b =1$ and $a \cdot c = 1$.  Then using the axioms for identity, associativity and commutativity:
$$b = 1 \cdot b = (a \cdot c) \cdot b = a \cdot (c \cdot b) = a \cdot (b \cdot c) = (a \cdot b) \cdot c =1 \cdot c= c.$$
\end{proof}

Given the uniqueness, we can define a \textit{partial} inverse operator $x^{-1}$ in any commutative ring  with 1.  We will write $\frac{1}{x}$ and $x^{-1}$ alternatively as synonyms.

\begin{definition}\label{division_by_inverse}
Let $R$ be a commutative ring with 1. Then it has a unique {\em partial division operator} defined by inverses 

(i) $ a \div b =_{def} a \cdot (b^{-1})$ where

(ii) if $ a\cdot c = 1$ then $a^{-1}= c$, and 

(iii) if for no $c \in R$, $ a \cdot c = 1$ then $a^{-1}$ is undefined.
\end{definition}

We will write $\frac{x}{y}$ and $x \div y$ alternatively as synonyms.

\begin{example}
{\em In this form of division we can calculate $a \div b$ whenever $b^{-1}$ exists. In $\Int$ this a limited form of division: the equation $b \cdot x = 1$ has solutions only when  $b = 1$ and $b = -1$.

In the finite ring $\Int_{n}$ there are more divisions. In  $\Int_{10}$, the equation $b \cdot x = 1 \ mod \ 10$ has solutions $b = 1, 3, 7$.}
\qed
\end{example}

\begin{definition}
Extending a ring $R$ with this operator $\div$ creates a {\em ring with partial division} denoted $R_{\div}$.
Extending a field $F$ with this operator $\div$ creates a {\em partial meadow} denoted $F_{\div}$.
\end{definition}

Let $\Sigma_r$ be a signature of rings and fields, and let $\Sigma_{r,\div}$ be a signature of rings and fields with division. Thus, $\Sigma_{r,\div}$ is the signature of a meadow.


\subsection{Totalising division in a ring}\label{division_in_a_ring}

To move from division as a partial operation to division as a total operation in a ring, we define common division:

\begin{definition}
Let  $R_{\div}$ be a commutative ring with partial division. Then, on adding the absorptive element $\bot$ to the domain of the ring and re-defining its division operation by

(iv) if $ a\cdot c = 1$ then $a^{-1}= c$, and 

(v) if for no $c \in R$, $ a \cdot c = 1$ then $a^{-1}= \bot$,

\noindent we have a {\em total} division operator; we call this total operation {\em common division}. 
\end{definition} 

Recalling the general method of extending an algebra with $\bot$, and specifically Definition \ref{totalising_partial_operations}, we will use the following notations: 

\begin{definition}
A  {\em commutative ring with inverse based common division} is a total algebra of the form $\mathsf{Enl}_\bot(R_{\div})$.  Let $\mathsf{CR}_{\div, \bot}$ be the class of all such algebras.
A {\em common meadow} is an algebra of the form $\mathsf{Enl}_\bot(F_{\div})$ wherein $F$ is field. Let $\mathsf{CM}$ be the class of all common meadows.
\end{definition}

Let $\Sigma_{r,\div, \bot}$ be a signature of rings and fields with common division wherein $\bot$ is a constant.

\begin{lemma}
Each common meadow is also a commutative ring with common division, but not every commutative ring with common division is a common meadow. In symbols, $\mathsf{CM} \subsetneqq \mathsf{CR}_{\div, \bot}$.
\end{lemma}

Here is a counter-example for the second clause of the lemma.

\begin{example}
{\em An example of a commutative ring with common division that is {\em not} a common meadow is $M =\mathsf{Enl}_\bot((\Int_n)_{\div})$ with $n = 10$. Notice that in $M$, $2 \cdot 3 = 6$. There is no $c$ such that $2 \cdot c = 1\mod 10$ and so the inverse $\frac{1}{2} = \bot$. It follows that there is a non-zero element without proper inverse. We further notice that, 
$$\frac{6}{2} =  6 \cdot \frac{1}{2} = 6 \cdot \bot = \bot \neq 3.$$  
On the other hand, as $3 \cdot 7 = 21 \mod 10 = 1 \mod 10$, we have 
$$\frac{6}{3} = 6 \cdot \frac{1}{3} = 6 \cdot  7  \mod 10 = 42 \mod 10 = 2.$$}
\qed
\end{example}

On forgetting division,  a commutative ring with partial division remains a ring and, similarly, a meadow whose division is partial remains a field. However, as noted earlier, the addition of the absorptive $\bot$ breaks the axioms and laws of rings. Thus, of immediate technical interest is the impact of $\bot$ on the familiar and fundamental properties of rings and fields. Indeed, what may be left after weakening the properties of commutative rings by adding $\bot$?


\subsection{Equational axioms}\label{the_equational_axiom_set}

Now, we present the axiom system $E_{\mathsf{ftc-cm}}$ that is at the centre of our programme on common meadows and a focus of this paper. The equations for $E_{\mathsf{ftc-cm}}$ split into two parts: 

(i) a set $E_{\mathsf{wcr},\bot}$ of equations for commutative rings now weakened by $\bot$, in Table \ref{EwcrBot}; and 

(ii) equations for division $\div$, in Table \ref{FTCcm}.

\begin{table}
\centering
\hrule
\begin{align}
	(x+y)+z 			&= x + (y + z)\\
	x+y     			&= y+x\\
	x+0     			&= x\\
	x + (-x) &= 0 \cdot x \\
	x \cdot (y \cdot z) 	&= (x \cdot y) \cdot z\\
	x \cdot y 			&= y \cdot x\\
	1 \cdot x 			&= x\\
	x \cdot (y+ z) 		&= (x \cdot y) + (x \cdot z) \\
	-(-x) 				&= x\\
	 0 \cdot (x \cdot x)	&= 0 \cdot x\\
	x + \bot 			&= \bot
\end{align}
\hrule
\medskip
\caption{$E_{\mathsf{wcr},\bot}$: equational axioms for weak commutative rings with $\bot$}
\label{EwcrBot}
\end{table}

\begin{table}
\centering
\hrule
\begin{align}
	{\tt import~}  &~ E_{\mathsf{wcr},\bot} \nonumber \\
	x				&= \frac{x}{1}\\
	\label{fracDivMult}
	\frac{x}{y}	\cdot \frac{u}{v}		&=  \frac{x \cdot u}{y \cdot v}\\
	\label{fracDivSum}		
	\frac{x}{y} + \frac{u}{v}	&= \frac{(x\cdot v) + (y \cdot u)}{y \cdot v}\\
	\frac{x}{y +( 0 \cdot z)} &= \frac{x + (0 \cdot z)}{y}\\
	\bot				&= \frac{1}{0}
\end{align}
\hrule
\medskip
\caption{$E_{\mathsf{ftc-cm}}$: Equational axioms for fracterm calculus for common meadows}
\label{FTCcm}
\end{table}

\begin{definition}
A model of $E_{\mathsf{ftc-cm}}$ we will call a {\em generalised common meadow}. We denote the class of such algebras $\mathsf{GCM} =\mathsf{Alg}(\Sigma_{r,\div, \bot}, E_{\mathsf{ftc-cm}})$.
\end{definition}

Because the generalised common meadows constitute a variety, the conditional equational theory -- and indeed the first order theory -- of the class of generalised common meadows is finitely axiomatized by $E_{\mathsf{ftc-cm}}$. 


\subsubsection{Fracterms and flattening}\label{fracterms}
An important feature of division is the introduction of fractional expressions. 
Now, in abstract data type theory, fractions can be given a clear formalisation as a 
syntactic object -- as a term over a signature containing division with a certain form. 
Rather than fraction we will speak of a \textit{fracterm}, 
following the terminology of~\cite{Bergstra2020} (item 25 of 4.2):

\begin{definition}
A {\em fracterm} is a term over $\Sigma_{r,\div}$ whose leading function symbol is division $\div$.   A {\em flat fracterm} is a fracterm with only one division operator.
\end{definition}

For reasons of convention and clarity, we will change notation from $\div$ to the familiar fraction notations when appropriate and so fracterms can have these forms 
$p \div q, \  \frac{p}{q}, \ p/q$
and flat fracterms have these forms in which $p$ and $q$ do not involve any occurrence of division. 
Here is an appropriate context for changing from $\div$ to the familiar fraction notation: 
the following simplification process is a fundamental property of working with fracterms.

\begin{theorem}\label{FF}
(Fracterm flattening \cite{BergstraP2015}.) 
For each term $t$ over $\Sigma_{r, \div, \bot}$ there exist $p$ and $q$ terms over $\Sigma_{r}$, i.e., both not involving $\bot$ or division, such that 
$$\displaystyle E_{\mathsf{ftc-cm}} \vdash t = \frac{p}{q},$$
 i.e., $t$ is provably equal to a flat fracterm. Furthermore, the transformation is computable.
\end{theorem}
\begin {proof}  
Immediate by structural induction on the structure of $t$, noting that any occurrence of $\bot$ can be replaced by $1/0$.  
\end{proof}

Thus, fracterms can be transformed uniformly into equivalent flat fracterms for use in any model of  $E_{\mathsf{ftc-cm}}$, i.e., in any generalised meadow. One of several key properties is that the axioms of $E_{\mathsf{ftc-cm}}$ guarantee flattening (Theorem \ref{FF}).

\subsubsection{A remark on (changing) terminology}

What we now call a generalised common meadow was originally called a common meadow in~\cite{BergstraP2015}. However, as applications in computing mainly deal with enlargements of fields with division we have chosen to use common meadow for a smaller class of structures where the non-$\bot$ elements constitute a field rather than for \textit{any} model of the equations in Tables~\ref{EwcrBot} and \ref{FTCcm}. For an arbitrary model of $E_{\mathsf{ftc-cm}}$ we prefer to use \textit{generalised common meadow}, or GCM, in an attempt to be concise and avoid overly logical jargon.
We mention that in~\cite{DiasD2023} and~\cite{DiasD2024} the original terminology of~\cite{BergstraP2015} is used.

\section{Commutative rings with $\bot$}\label{rings_with_bot}

We have developed these axioms $E_\mathsf{ftc-cm}$ to make a theory of common meadows. How much lifts from fields to commutative rings with common division? The following observation is due to Dias and Dinis in~\cite{DiasD2023}. 


\begin{proposition}
\label{RingSound}
Each commutative ring with inverse based common division satisfies all the equational axioms in $E_{\mathsf{ftc-cm}}$ for common meadows.
\end{proposition}
\begin{proof}
This is a matter of straightforward inspection of all equations of $E_{\mathsf{ftc-cm}}$.
\end{proof}

\begin{example}
{\em There are models of $E_{\mathsf{ftc-cm}}$ that are commutative rings and are not meadows: starting with a ring that is not a field,  e.g. $\Int$, we take the $\bot$-enlargement of it and then expand the resulting structure with inverse based common division to obtain a generalised common 
meadow that is not a common meadow.

There are are models of $E_{\mathsf{ftc-cm}}$ that are not commutative rings: add a new constant $c$ to the signature to make $\Sigma_{r,\div, \bot,c}$ and do not add any equations mentioning $c$ to $E_{\mathsf{ftc-cm}})$. Consider the class $\mathsf{Alg}(\Sigma_{r,\div, \bot,c}, E_{\mathsf{ftc-cm}}$.
It is easy to see that an initial algebra $A$ of this equational class is a generalised common meadow, as it satisfies $E_{\mathsf{ftc-cm}}$; but $A$ is not a ring as neither $c \cdot 0 = 0$ nor $c = \bot$ can be shown from $E_{\mathsf{ftc-cm}}$, so that neither identity is true in $A$. }
\qed
\end{example}

Consider the conditional formula -- which is not a conditional equation -- taken from~\cite{BergstraP2015}:

\begin{definition}
The following formulae is termed the {\em Normal Value Law}
$$\mathsf{NVL}: \ \  x \neq \bot \to 0 \cdot x = 0$$
or, equivalently,
$$\mathsf{NVL}: \ \  0 \cdot x \neq 0 \to   x = \bot.$$
\end{definition}

\begin{proposition} 
The class $\mathsf{CR}_{\div, \bot}$  of commutative rings with inverse based common division are precisely the models of $E_{\mathsf{ftc-cm}}$ that satisfy
$\mathsf{NVL}$.
\end{proposition}
\label{axiom_rings_with_IBCD}

\begin{proof} That a commutative ring with inverse based common division satisfies $\mathsf{NVL}$ is immediate because in a ring all elements $x$ satisfy $0 \cdot x = 0$.
For the other direction, suppose $M \models E_{\mathsf{ftc-cm}} + \mathsf{NVL}$. 
The non-$\bot$ elements of $M$ satisfy $0\cdot x = 0$ so that $M$ is an expansion of $\mathsf{Enl}_\bot(R)$ for a ring $R$. 
Now if $\frac{1}{p}$ is non-$\bot$ for $p \in M$ then $0 \cdot \frac{1}{p} = 0$ and 
we find that 
$$\frac{p}{p} = 1 + \frac{0}{p}= 1 + 0 \cdot \frac{1}{p}= 1 + 0 = 1$$
 so that $\frac{1}{p}$ is defined with the same value in the enlargement of $R$ to a ring with inverse based common inverse. 
\end{proof}

We will now consider a conditional equation -- taken from~\cite{BergstraP2015}:

\begin{definition}
The following conditional equation is termed the {\em Additional Value Law}
$$\mathsf{AVL}: \ \  \frac{1}{x} = \bot \to 0 \cdot x = x.$$
\end{definition}

\begin{proposition}\label{axiom _meadows_relative_to_rings}
Within the class $\mathsf{CR}_{\div, \bot}$ of commutative rings with inverse based common division, the common meadows are precisely those which satisfy the conditional equation 
$\mathsf{AVL}$.
\end{proposition}

\begin{proof}
$\mathsf{AVL}$ is satisfied in any common meadow because $ \frac{1}{x} = \bot$ is satisfied only by $x=0$ and $x=\bot$, 
and in both cases $0\cdot x = x$. 
On the other hand, contrapositively, let $R$ be a ring which is not a field, then for some non-zero $a \in R$ it 
must be that there is no proper inverse for it, i.e., there is no $c \in R$ with $a \cdot c = 1$. Thus, in 
$\mathsf{Enl}_\bot(R_{\div})$ it is the case that $\frac{1}{a} = \bot$. However, as $a \in R$, $0 \cdot a = 0$,  and $0 \cdot a \neq a$ as $a \neq 0$ so $\mathsf{AVL}$ is not satisfied.
\end{proof}

As an immediate consequence we find:

\begin{proposition}
\label{incomplete_conditional_theory} 
The axiom system $E_{\mathsf{ftc-cm}}$ for common meadows is {\em not} 
complete for the conditional equational theory of common meadows.
\end{proposition}

\begin{proof} 
$\mathsf{AVL}$ is true in all common meadows but it is invalid in any generalised
common meadow $M$ obtained by enlarging a ring that is not a field with $\bot$ and inverse 
based common division. As $M  \models E_{\mathsf{ftc-cm}}$ (by Proposition~\ref{RingSound}) 
and $M  \not\models \mathsf{AVL}$, by soundness, we find that $E_{\mathsf{ftc-cm}} \not \vdash \mathsf{AVL}$. 
\end{proof} 
The result of Proposition~\ref{incomplete_conditional_theory} was obtained in~\cite{BergstraT2022CJ} with a different proof.
Note, too, that the conditional equation 
$$x \cdot x = 0 \to x = 0$$ 
is valid in all fields and in all common meadows. But it is not valid in all rings and for that reason also not valid in all rings enlarged with inverse based common division. 

Given the significance of $\mathsf{AVL}$ (Proposition \ref{axiom _meadows_relative_to_rings}), and its independence (Proposition \ref{incomplete_conditional_theory}), we examine $E_{\mathsf{ftc-cm}} + \mathsf{AVL}$:

\begin{proposition}
$E_{\mathsf{ftc-cm}} + \mathsf{AVL} \vdash x \cdot x = 0 \to x = 0$.
\end{proposition}

\begin{proof} From~\cite{BergstraP2015}, we have $E_{\mathsf{ftc-cm}} \vdash 0 \cdot x \cdot x  = 0  \to 0 \cdot x = 0$. Moreover,
 $E_{\mathsf{ftc-cm}} \vdash \frac{1}{x} = \frac{x}{x \cdot x}$. Thus from $x \cdot x = 0$ we find $\frac{1} {x} = \frac{x}{x \cdot x} = 
 \frac{x}{0} = x \cdot \bot = \bot$ so that with $\mathsf{AVL}$ we have $0 \cdot x = x$. Taking both facts together we find $x = 0 \cdot x = 0$.
\end{proof}

We will return to these issues later. In contrast to Proposition \ref{incomplete_conditional_theory}, we have shown previously:

\begin{theorem}
\label{ComplForRings}
The equational axiom system $E_{\mathsf{ftc-cm}}$ for common meadows is complete for the equational theory of  commutative rings with inverse based common division. For every equation $e$,
$$\displaystyle E_{\mathsf{ftc-cm}} \vdash e \  \textrm{if, and only if,}  \ \mathsf{CR}_{\div, \bot} \models e.$$
 
\end{theorem}
\begin{proof} Because all common meadows are commutative rings with common division, the equational theory of commutative rings with inverse based common division is contained in the equational theory of common meadows -- recall Lemma \ref{subclass_lemma}.

Therefore, completeness of $E_{\mathsf{ftc-cm}}$ for the equational theory of commutative rings with inverse based common division follows immediately from the completeness of $E_{\mathsf{ftc-cm}}$ for the equational theory of common meadows, which has been demonstrated in~\cite{BergstraT2023arxiv}.
\end{proof}

In addition, Proposition \ref{incomplete_conditional_theory} draws attention to the open Problem \ref{MainProblem} of the Introduction.


\subsection{Axiomatising the class of common meadows}
The following facts are related to general results  in~\cite{BergstraT2022TM}, recalled here as Theorem \ref{NonQ-Variety}. 
More specifically, the following fact follows from  Theorem \ref{NonQ-Variety} about 
$\bot$-enlargements.

\begin{proposition}
The class $\mathsf{CM}$ of common meadows is not a quasivariety. The class $\mathsf{CR}_{\div, \bot}$ of commutative rings with inverse based common division is not a quasivariety.
\end{proposition}

Here is a direct proof using a technique with initial algebras:
\begin{proof} 
Consider the second statement.
Suppose that conditional equational theory $E_{ce}$ has precisely the commutative rings with inverse based common division as its models, i.e., $\mathsf{CR}_{\div, \bot} = \mathsf{Alg}(\Sigma, E_{ce})$.

We introduce fresh constants $c$ and $d$. We consider an initial algebra of $E' = E \cup \{ c+ d = \bot\}$. Now there is a 
commutative ring with inverse based common division $M_1$ that satisfies $E'$ by taking $c= 0$ and $d = \bot$ from which it follows  that $E' \not \vdash c = \bot$.
Similarly, there is a model $M_2$ of $E'$ in which $d=0$ and $c = \bot$  so that $E' \not \vdash d= \bot$. It follows that, in an initial algebra of $E'$, $c+d = \bot$ while $c \neq \bot$ and $d\neq \bot$. However, this latter situation cannot occur in any structure of the form $\mathsf{Enl}_\bot(R_{\div})$ because $c \in R$ and $d \in R$, whence $c+d \in R$, thereby contradicting $c+ d = \bot$. 
\end{proof}

This basic fact may be contrasted with Theorem~\ref{ComplForRings} above. Indeed, while $E_{\mathsf{ftc-cm}}$ completely axiomatizes all equations true in all common meadows it fails, i.e.,  it is too weak, to define the class of common meadows. Common meadows can be defined with a first order theory but not with equations, or with conditional equations alone.

\begin{proposition}
There is a consistent conditional equational extension $T$ of $E_{\mathsf{ftc-cm}}$ which has no commutative ring with inverse based common division as a model.
\end{proposition} 

\begin{proof}
We set  
$$T = E_{\mathsf{ftc-cm}}~\cup \{0 \cdot \frac{1}{1+1} = 0 \to 0=1, \frac{1}{1+1} = \bot \to 0=1\}.$$
Let $M$ be an initial algebra of $E_{\mathsf{ftc-cm}}$. Then $0 \cdot \frac{1}{1+1} = 0$ is not valid in $M$ because $E_{\mathsf{ftc-cm}}$ allows a model (based on the prime field $\Int_{2}$ of characteristic $2$) in which 
$$0 \cdot \frac{1}{1+1} = 0 \cdot \bot = \bot \neq 0$$
and $\frac{1}{1+1} = \bot$ 
is not valid because that equation fails in a common meadow of rational numbers. It follows that $M$ satisfies $T$. 

Now consider any commutative ring with inverse based common division, say $M'= \mathsf{Enl}_\bot(R_\div)$ which is a model of $T$. In $M'$, 
$0 \cdot \frac{1}{1+1} \neq 0$ so that $ \frac{1}{1+1} \notin R$ and by definition of $\mathsf{Enl}_\bot(R_\div)$ 
it must be the case that $\frac{1}{1+1} = \bot$ in $M'$ which contradicts $M' \models T$.
\end{proof}


\section{GCM homomorphisms}\label{homomorphisms}
GCMs allow a rich world of homomophisms which has first been explored in general by Dias and Dinis in~\cite{DiasD2023}, 
and in the context of finite GCMs in~\cite{DiasD2024}, where however, homomorphisms  are not required to respect division. Below we will insist that a GCM-homomorphism respects division as well. We start our discussion of GCM homomorphisms with two examples.
\begin{example} 
{\em The ring homomorphism $\phi: \Int_4 \to \Int_2$ extends to a homomorphism from 
$\bar{\phi}: \mathsf{Enl}_\bot((\Int_4)_\div) \to \mathsf{Enl}_\bot((\Int_2)_\div)$.

However, the ring homomorphism $\phi: \Int_6 \to \Int_3$ does not extend to a homomorphism from 
$\mathsf{Enl}_\bot((\Int_6)_\div)$ to $\mathsf{Enl}_\bot((\Int_3)_\div)$. To see this notice that in $\mathsf{Enl}_\bot((\Int_6)_\div)$ the inverse of $2$ is $\bot$ while in 
$\mathsf{Enl}_\bot((\Int_3)_\div)$ the inverse of $2$ is $2$.}
\qed
\end{example}

\begin{proposition} 
\label{HomToRwCD}
Every generalised common meadow (GCM) allows a homomorphism $\rho$ to a commutative ring with inverse common division.
\end{proposition}

\begin{proof}
Let $M$ be a GCM, and using the notation of~\cite{DiasD2023} we write $M_0 = \{a \in M \mid 0 \cdot a = 0\}$.   $M_0$ is a commutative ring to which we add inverse common division, $\mathsf{Enl}_{\bot}(M_0)$.

We define $\rho$ as follows: 
$$\rho(a) = a \ \textrm{for} \ a \in M_0  \  \textrm{and} \ \rho(a) = \bot \ \textrm{for} \ a \not\in M_0.$$ 
To see that $\rho$ is a homomorphism: 

\noindent (i) \textit{Constants}. as $0,1 \in M_0$  and $\bot \notin M_0$, $\rho$ preserves the constants. 

\noindent (ii) \textit{Addition}. To check $\rho(a) + \rho(b) = \rho (a+ b)$ we have four cases: 

(ii1) if $a \in M_0$ and $b \in M_0$ then $0 \cdot (a+b) = 0 \cdot a + 0 \cdot b = 0 + 0 = 0$ so that  $\rho(a) + \rho(b) = a + b = \rho (a+ b)$; 

(ii2) $a \in M_0$ and $b \notin M_0$ then $\rho(a) + \rho(b) = a + \bot = \bot =  \rho (a+ b)$, while assuming that $a + b \in M_0$ and if $a + b \in M_0$ we find
a contradiction as follows $0 \cdot b = 0 \cdot 0 + b = 0 \cdot a + 0 \cdot b = 0 \cdot (a + b) = 0$, i.e. $b \in M_0$; 

(ii3) $a \notin M_0$ and $b \notin M_0$, is dealt with as case (ii2); and finally 

(ii4) $a \notin M_0$ and $b \notin M_0$, from Proposition 2.3.1. (ce5) in~\cite{BergstraP2015} we find that $a+b \in M_0$ implies $a \in M_0$ so that we know that  $a+b \notin M_0$ whence $\rho(a) + \rho(b) = \bot + \bot = \bot =  \rho (a+ b)$.

\noindent (iii) \textit{Multiplication}. This case works the same as addition.

\noindent (iv) \textit{Additive inverse}. This is immediate.

\noindent (v) \textit{Multiplicative inverse}. This has two subcases: (v1) $a \in M_0$ and (v2) $a \notin M_0$.

For (v1), if $a \in M_0$ there are two subcases: (v1i) $ a^{-1} \in M_0$, which is immediate and (v1ii)  $ a^{-1} \notin M_0$. In case (v1ii) it cannot be the case that $ a \cdot a^{-1} = 1$ as otherwise 
$ 0 \cdot a^-1= (0 \cdot a) \cdot a^{-1} = 0 \cdot (a \cdot a^{-1}) = 0 \cdot 1 = 0$ thereby contradicting  $ a^{-1} \notin M_0$ so we find that $\rho (a^{-1}) = \bot$ by definition of $\rho$ in a common meadow and $\rho(a)^{-1} = \bot$ by definition of a vulgar meadow. In case 

For (v2), i.e., $a \notin M_0$ we find that $a^{-1} \notin M_0$ because otherwise $ 0 \cdot a= (0 \cdot a^{-1}) \cdot a = 0 \cdot (a^{-1} \cdot a) = 0 \cdot 1 = 0$.
\end{proof} 


\subsection{GCM homomorphisms and $\mathsf{AVL}$}\label{homomorphisms_and_AVL}

We first notice that $\mathsf{AVL}$ is a natural property for GCM's.

\begin{proposition}
\label{InitAlgAVL}
Let $I(\Sigma_{\mathsf{cm}},E_{\mathsf{ftc-cm}})$ be an initial algebra for the equations $E_{\mathsf{ftc-cm}}$. Then $I(\Sigma_{\mathsf{cm}},E_{\mathsf{ftc-cm}})$ satisfies 
$\mathsf{AVL} $. 
\end{proposition}
\begin{proof}
To see this, we consider a closed term $t$ such that $\frac{1}{t} = \bot$ is provable from $E_{\mathsf{ftc-cm}}$.
Using fracterm flattening (Theorem \ref{FF}) we find that $t = \frac{n}{m}$ for appropriate numerals $n$ and $m$. 
Then $\frac{1}{t} = \frac{m \cdot m}{n \cdot m}$. Now $\frac{m \cdot m}{n \cdot m}= \bot$ implies $n \cdot m = 0$
from which it follows that either $n= 0$ or $m=0$ and in both cases $0 \cdot t = t$.
\end{proof}

Now, 
$\rho(M) \cong \mathsf{Enl}_\bot(\Int_\div)$ and so fails $\mathsf{AVL} $, so that $\mathsf{AVL}$ is not 
preserved under homomorphisms of the form $\rho$. 

We will now consider a different kind of homomorphism on GCM's.
Part (i) of the following observation  is made in~\cite{DiasD2023}.

\begin{lemma} 
\label{HomToGCM1}
Let $M$ be a GCM and let $a  \in M$ be such that $a \neq \bot$, $a \neq 0$ and $0 \cdot a = a$.
Let the function $\phi_a \colon M \to M$ be defined by 
$$\phi_a(b)= b + a,$$ and write $x =_a y$ for $\phi_a(x) = \phi_a(y)$:

(i) $\phi_a$ is a homomorphism, the image $\phi_{a}(M)$ of which is again a non-trivial GCM (i.e., $\phi_a(0) \neq \phi_a(1)$); 

(ii) if $M$ satisfies $\mathsf{AVL}$ then the homomorphic image $\phi_{a}(M)$ satisfies $\mathsf{AVL}$;

(iii) if $a = \frac{1}{b}$ then $a \cdot b =_a 1$.
\end{lemma}

\begin{proof} 
\label{HomDandD}
For (i), we refer to~\cite{DiasD2023}, though the check that 
$p =_a q \iff p + 0 \cdot a = q + 0 \cdot a$ is a congruence is immediate for all cases. 

For (ii), consider $c \in M$ such that $\phi_a(\frac{1}{c}) = \bot$, it must be verified that 
$0 \cdot \phi_a(c) = \phi_a(c)$. Now, $\phi_a(\frac{1}{c}) = \bot$ implies that (working in $M$, 
and doing some of the work of the proof of (i)) 
$$\frac{1}{c} + 0 \cdot a = \frac{1}{c} + a = \bot + a = \bot$$ 
and
$$\frac{1}{c} + 0 \cdot a =  \frac{1 + 0 \cdot a}{c} = \frac{1}{c + 0 \cdot a} = \frac{1}{\phi_a(c)} $$ 
so that $\frac{1}{\phi_a(c)} = \bot$ and with $\mathsf{AVL}$ in $M$ one finds 
$0 \cdot \phi_a(c) = \phi_a(c) $ as required.

For (iii): $a \cdot b = b \cdot \frac{1}{b} = 1 + \frac{0}{b} = 1 + 0 \cdot a =_a 1$.

(iv) is immediate.
\end{proof}

Given a GCM $M$ with non-$\bot$ $a \in M$ such that $0 \neq 0 \cdot a$ one may define a congruence $='_a$ as follows:

(i) $b ='_a \bot \iff b + 0 \cdot a = b$, and 

(ii) $b ='_a c \iff (b = c\, \vee\, (b='_a \bot \wedge c ='_a \bot))$. 

It is a straightforward verification that $='_a$ is a congruence relation with $1 \not ='_a \bot$ and consequently $1 \not ='_a 0$. $M/\! ='_a $ is a non-trivial GCM. Note the natural  homomorphism: 

\begin{proposition} 
\label{CounterExAVL}
The natural homomorphism $\pi^\bot_{a}: M \to M/\! ='_a$ need not preserve the validity of $\mathsf{AVL}$.
\end{proposition}

\begin{proof}
Let $a$ and $b$ be fresh constants and let $I_S$ be an initial algebra of the algebraic specification
$$S=(\Sigma_{\mathsf{cm}}\cup\{a,b\},E_{\mathsf{ftc-cm}} \cup \{\mathsf{AVL},\frac{1}{b} + 0 \cdot a = \frac{1}{b}\}).$$  Now, $I_S$ satisfies $\mathsf{AVL}$ while $I_S/\!='_a $ does not. 

To see the this, we notice that in $I_S/\!='_a $, $\frac{1}{b} = \bot$, i.e., 
$\frac{1}{b} ='_a  \bot$ while $0 \cdot b \not ='_a b$. To see that $0 \cdot b \not ='_a b$ fails in $I_S/\!='_a $ first notice that upon choosing $\frac{1}{2}$ for $a$ and $\bot$ for $b$, $I_S$ is mapped to an initial algebra $M'$ of the specification 
$S' =(\Sigma_{\mathsf{cm}}\cup\{a,b\},E_{\mathsf{ftc-cm}}\cup \{\mathsf{AVL},b = \frac{1}{2}, a = \bot \})$ which is also an initial algebra of $S' =(\Sigma_{\mathsf{cm}}\cup\{a,b\},E_{\mathsf{ftc-cm}}\cup \{b = \frac{1}{2}, a = \bot \})$ in view of Proposition~\ref{InitAlgAVL}.
Now $M'$ satisfies $ 0 \cdot b \neq 0$ (i.e., $0 \cdot \frac{1}{2} \neq 0 \cdot \frac{1}{2}$) and (with $b = \frac{1}{2}$, $a = \bot$) $b+ 0 \cdot a \neq b$, and $0 \cdot b + 0 \cdot a \neq 0 \cdot b$. It follows that none of $ 0 \cdot b = 0$, $ b + 0 \cdot a = b$ and $0 \cdot b + 0 \cdot a = 0 \cdot b$ are provable from $S$ so that neither of these equations is true in $M$, with as a consequence that in $M$ $0 \cdot b \neq b$, thereby refuting the validity of  $0 \cdot b = b$ in $I_S/\!='_a $ and for that reason $\mathsf{AVL}$ refuting as well.
\end{proof}
The example of Proposition~\ref{CounterExAVL} indicates that chains of homomorphisms 
applied to a GCM in order to get closer to a common meadow are best made up from 
homomorphisms of the type $\phi_{0 \cdot a}(b) = b + 0 \cdot a$, rather than $\rho$ or $\pi^{\bot}_{a}$.

Recall completeness for conditional equations (Definition \ref{def:completeness}), 
and that $E_{\mathsf{ftc-cm}}$ does not axiomatise the conditional theory of common meadows -- Proposition \ref{incomplete_conditional_theory}.
However, it does more than axiomatise the equational theory of commutative rings -- Proposition \ref{ComplForRings}: 

\begin{theorem}\label{completeness_for_rings} 
$E_{\mathsf{ftc-cm}}$ provides a complete equational axiomatisation of the conditional equational theory of the class of commutative rings with inverse based common division.
\end{theorem}

\begin{proof} We need to show that whenever a conditional equation $\psi$ is true in all commutative rings with inverse common division it is derivable from $E_{\mathsf{ftc-cm}}$, i.e., $\psi$ is true in all GCMs. So suppose that $\psi$ is refuted in a GCM, say $M$. We will transform $M$ into a commutative ring with common division $M'$ which refutes $\psi$ as well. We consider three cases.

(i) We write $\psi \equiv t_1 = r_1 \wedge \ldots \wedge t_n = r_n \to t=r$, and we assume that variables $x_1,\ldots,x_k$ occur in $\psi$
and that substituting $a_1, \ldots,a_k$ for these variables provides an instance in $M$ refuting $\psi$. Let $b = t(a_1,\ldots ,a_k)$ and
$c = r(a_1,\ldots ,a_k)$. Because $b\neq c$ at least one of both, say $b$ differs from $\bot$. 

We assume first that $c = \bot$ and $b \neq \bot$. We know that $0 \cdot b \neq \bot$ and we consider the mapping $\rho \circ \phi_{0 \cdot b}$ with $\rho$ as used in the proof of Proposition~\ref{HomToRwCD} and $\phi_{0 \cdot b}$ as used in the proof of Lemma~\ref{HomToGCM1}.
As $\rho \circ \phi_{0 \cdot b}$ is a homomorphism it preserves all conditions of the conditional equation $\phi$. 

Concerning the conclusion of $\psi$ we notice that 
$$\phi_{0 \cdot b}(b) = b+ 0 \cdot b = (1 + 0 )\cdot b = 1 \cdot b = b,$$
and in the image $\phi_{0 \cdot b}(M)$, $\phi_{0 \cdot b}(b) \neq \bot$ because if $ b =_{0 \cdot b} \bot$ so that $b + 0\cdot b = b = \bot$ contradicting the assumption on $b$ and $c$. Clearly, $\phi_{0 \cdot b}(c) = \phi_{0 \cdot b}(\bot)= \bot$ so that the conclusion of $\psi$ remains false in the image of $\phi_{0 \cdot b}$  on the substitution
 $\phi_{0 \cdot b}(a_1),\ldots, \phi_{0 \cdot b}(a_k)$.  In $\phi_{0 \cdot b}(M)$ then, $0 = 0 \cdot b$ so that $\rho$ acts as identity on $\phi_{0 \cdot b}(b)$, with the effect that $\psi$ fails in $\rho(\phi_{0 \cdot b}(M))$, which is a commutative ring with inverse based common division.

(ii) Next, we consider the case that $b$ and $c$ are both non-bot and that $0 \cdot b = 0 \cdot c$. 
Again, we may consider the homomorphism $\rho \circ \phi_{0 \cdot b}$ and we find $\rho (\phi_{0 \cdot b}(b)) = b$ as above and 
$\rho (\phi_{0 \cdot b}(c) )=  c$ so that again in $\rho \circ \phi_{0 \cdot b}(M)$ the conditional equation $\psi$ is invalid on the substitution $\phi_{0 \cdot b}(a_1),\ldots, \phi_{0 \cdot b}(a_k)$.

(iii) Finally, we assume that $0 \cdot b \neq 0 \cdot c$. Suppose that both $(0 \cdot b) \cdot (0 \cdot c) = 0 \cdot b$ and $(0 \cdot c) \cdot (0 \cdot c) = 0 \cdot c$ then $0 \cdot b = 0 \cdot c$ which contradicts the assumptions on $b$ and $c$. We assume w.o.l.g. that  
$(0 \cdot b) \cdot (0 \cdot c) \neq 0 \cdot b$. Again we will use the homomorphism $\rho \circ \phi_{0 \cdot b}$. Now we find that 
$0 \cdot \psi_{0 \cdot b}(b) = \phi_{0 \cdot b}(0)$ while $0 \cdot \phi_{0 \cdot b}(c) \neq \phi_{0 \cdot b}(0)$, as otherwise 
 $ 0 \cdot (0 \cdot c) =_{0 \cdot b} 0$ which would imply that $(0 \cdot b) \cdot (0 \cdot c) \neq 0 \cdot b$. It follows that $\rho$ will leave
 $\phi_{0 \cdot b}(b)$ unchanged while $\rho(\phi_{0 \cdot b}(c)) = \bot$ with the effect that $\phi$ fails in $\rho (\phi_{0 \cdot b}(M))$.
\end{proof}


\subsection{$\bot$-splitting in GCMs}

\begin{definition}
A GCM $M$ features (non-trivial) $\bot$-{\em splitting} if $\bot$ is the sum of two non-$\bot$ elements in $M$. 
\end{definition}

Clearly, a common meadow will not feature non-trivial $\bot$-splitting and, more generally, nor will a commutative ring with inverse based common division.

Moreover, if in a GCM $M$, $a, b \in M$ feature non-trivial $\bot$-splitting then $0 \cdot a \neq 0$ as otherwise $0 \cdot (a+ b) = 0 \cdot \bot$ so that $0 \cdot b = \bot$ and then $b = \bot$; similarly, it follows that $0 \cdot b \neq 0$.

\begin{example}
{\em Let $a$ and $b$ be two constants with associated axiom $a + b= \bot $. In the initial algebra $I$ of $E_{\mathsf{ftc-cm}} + \{ a + b= \bot \} $ we find $0 \neq a \neq \bot$, $0 \neq b \neq \bot$ and $ a\neq b$. Thus, $I$ allows non-trivial $\bot$-splitting.} 
\qed
\end{example}

\begin{proposition}
Let $M$ be a GCM which satisfies $\mathsf{AVL}$.  If $M$  possesses non-trivial zero divisors then $M$ features non-trivial $\bot$-splitting.
\end{proposition}

\begin{proof}
Suppose $M$ allows non-trivial zero divisors, then for some $a,b \in M$, $ a \cdot b = 0$ while $a \neq 0 \neq b$. 
Now, consider the sumterm $t = \frac{1}{a} + \frac{1}{b}$. We will find that $t$ features non-trivial $\bot$-splitting. We have $t = \frac{a + b}{a \cdot b} = \frac{a + b}{0}= \bot$. Moreover, if $\frac{1}{a} = \bot$ then with $\mathsf{AVL}$, $0 \cdot a = a$ and from $ a \cdot b = 0$ follows $ 0 \cdot a \cdot b = 0$, which implies $0 \cdot a = 0$. Both facts together yield $a = 0$, which contradicts the assumptions on $a$ whence $\frac{1}{a} \neq \bot$. Similarly, it follows that $\frac{1}{b} \neq \bot$.
\end{proof}

If $a$ and $b$ feature $\bot$-splitting the homomorphisms $\phi_{0 \cdot a}$ and $\phi_{0 \cdot b}$ do not commute.
Indeed we find that $b =_{0 \cdot a} \bot$ and that $a =_{0 \cdot b} \bot$. It follows that the union of $=_{0 \cdot a}$ and $=_{0 \cdot b}$ 
does not generate a congruence relation and that the transitions from $M$ to $M/\!=_{0 \cdot a}$ and to $M/\!=_b$ 
exclude one-another. Moreover, if for non-$\bot$ $b$ it is the case that $b =_{0 \cdot a} \bot$ then 
$a$ and $b$ feature non-trivial $\bot$-splitting in $M$. Indeed: if $b + 0\cdot a = \bot + 0 \cdot a$ then $b + 0\cdot a = \bot$ and 
$0 \cdot (b + 0\cdot a) =0 \cdot \bot$ so that $0 \cdot (a+b) = \bot$ and $a + b = \bot$.

\begin{proposition} 
Given a GCM $M$ which satisfies $\mathsf{AVL}$ and which avoids non-trivial $\bot$-splitting, 
the union  of the congruences $=_{0 \cdot a}$ for all non-$\bot$ elements $a$ of $M$ 
generates a non-trivial congruence $=_{\mathsf{cm}}$ on $M$ such that $M/\!=_{\mathsf{cm}}$ is a common meadow.
\end{proposition}

\begin{proof} In view of the absence of $\bot$-splitting,  
the union of the non-trivial congruences $=_{0 \cdot a}$ and $=_{0 \cdot b}$ generates the non-trivial congruence 
$=_{0 \cdot (a+b)}$. Further $=_{\mathsf{cm}}$ is a commutative ring with inverse based common division, 
as any $a \in M$ with $0 \cdot a \neq 0$ will satisfy $0 \cdot a =_{\mathsf{0\cdot a}} 0$ and for that reason also 
$0 \cdot a =_{\mathsf{cm}} 0$. 

To see that $M/\!=_{\mathsf{cm}}$ satisfies $\mathsf{AVL}$ consider $a \in M$ such that 
$\frac{1}{a} =_{\mathsf{cm}} \bot$. Then for some non-$\bot$, $b \in |M|$, $\frac{1}{a} + 0 \cdot b=_{\mathsf{cm}} \bot + 0\cdot b= \bot$.
Because $M$ avoids $\bot$-splitting and $0 \cdot b \neq \bot$, $\frac{1}{a} = \bot$ so that with $\mathsf{AVL}$, in $M$, $a = 0 \cdot a$, and therefore $0 \cdot a =_{\mathsf{cm}} a$ as required for the check of $\mathsf{AVL}$.
\end{proof}

In case $M$ allows non-trivial $\bot$-splitting, choices must be made which element of a $\bot$-splitting pair to map to $0$ and which on to map to $\bot$ as the ordering of factorization now matters. Now, it is possible to use a  a well-ordering of the non-$\bot$-elements in order to remove all $\bot$-splitting pairs.

\begin{lemma} 
\label{HomToCM2}
Given a  GCM $M$ that satisfies $\mathsf{AVL}$, with a non-$\bot$ element 
$a \in M$ with $0 \neq 0 \cdot a$,  then there is a congruence $=_{\mathsf{cm}}$ on $M$ 
such that  $M/=_{\mathsf{cm}}$ is a common meadow and $a \not =_{\mathsf{cm}} \bot$. 
\end{lemma}

\begin{proof}
Given $M$ with properties as stated  we first define the set $U_M$ by 
$$U_M = \{b \in |M| \mid 0 \neq  b\, \&\, 0 \cdot b = b\, \&\, b \neq \bot \}.$$
Now, let $\alpha$ be an ordinal with $g$ a 1-1 mapping from 
the successor ordinals in $\alpha$ to $U_M$ such that $g(1) = a$. The congruences $=^g_\beta$ for $\beta \leq \alpha$ are as follows:
\begin{itemize}
\item $=^g_0$ is the equality relation of $M$,

\item for a successor ordinal $\beta  +1$: if $0 \cdot g(\beta+1) =^g_{\beta}0$ or if 
$0 \cdot g(\beta+1) =^g_{\beta}\bot$ then $=^g_{\beta+1}$ is the same as $=^g_{\beta}$; otherwise
$=^g_{\beta+1}$ is the congruence relation generated by the union of $=^g_{\beta}$  and 
$=_{g(\beta+1)}$, and

\item
for a limit ordinal $\beta$,  $=^g_{\beta}\, = \bigcup_{\gamma < \beta} =^g_{\gamma}$.
\end{itemize}
\noindent 
Now, we will find that $\equiv^g_{\alpha}$ is a nontrivial congruence on $M$ such that  in $M/\!\equiv^g_{\alpha}$ for every $c$ either $c = \bot$ or $0 \cdot c = 0$, so that $M/\!\equiv^g_{\alpha}$ is a commutative ring with inverse based common division, 
which satisfies $\mathsf{AVL}$ (in view of Proposition~\ref{HomToGCM1} each factor GCM $M/=^g_{\beta}$ satisfies $\mathsf{AVL}$),
and which is a common meadow in view of Proposition~\ref{axiom _meadows_relative_to_rings}. Moreover, $a \not = \bot$ as in $M/=_{g(1)} = M/=_a$, $a = 0$ and not $ a = \bot$, and all further congruences are consistent extensions of $=_a$.
\end{proof}


\subsection{Weak commutative rings with $\bot$}

Consider the effect of $\bot$ on the ring properties.

\begin{definition}
An algebra $R$  satisfying the equations 1-10 of Table~\ref{EwcrBot} is called a {\em weak commutative ring} ($WCR$). An algebra $R_\bot$  satisfying all the equations 1-11 of Table \ref{EwcrBot} is called a {\em weak commutative ring with $\bot$}  ($WCR_{\bot}$).
\end{definition}

Every commutative ring is a weak commutative ring but not conversely.

\begin{lemma}
 For any given commutative ring $R$, $R_\bot$ is a $WCR_{\bot}$.  Further,  a WCR of the form $R_\bot$ satisfies $\mathsf{NVL}$
$$0 \cdot x \neq 0 \to x = \bot,$$
and all WCR's that satisfy 
this formula are of the form $R_\bot$.
\end{lemma}

Let $\Sigma_{r,\bot}$ be a signature of rings and fields with $\bot$.  Now, given a WCR $R_\bot$ the algebra $(R_\bot)_\div$ is obtained by expanding $R_\bot$ with a division function $x \div y$ such that $x \div y = x \cdot (1 \div y)$ and $1 \div a$ is $b$ if $ a \cdot b = 1$ (such $b$ is unique, see Lemma \ref{dividers_are_well-defined} and  Definition \ref{division_is_well-defined} above), and 
$1 \div a$ is undefined if no proper inverse $b$ is available for $a$ in $R_\bot$. 
We notice that if $R_\bot$ is non-trivial, then $(R_\bot)_\div$ is a partial algebra.

Although $\bot \in R_\bot$ already,  for $(R_\bot)_\div$ we totalise and define $\mathsf{Enl}_\bot((R_\bot)_\div)$ by having $a \div b = \bot$ whenever $a \div b$ is undefined in
$(R_\bot)_\div$.

\begin{proposition}
$\mathsf{Enl}_\bot((R_\bot)_\div)$ satisfies $E_{\mathsf{ftc-cm}}$ if, and only if, $R_\bot$ satisfies $\mathsf{NVL}$: $0 \cdot x \neq 0 \to x = \bot$.
\end{proposition}
\begin{proof} 
If $R_\bot$ satisfies $\mathsf{NVL}$
then it is of the form $\mathsf{Enl}_\bot (R)$ for some ring $R$, 
so that $\mathsf{Enl}_\bot((R_\bot)_\div)$ satisfies $E_{\mathsf{ftc-cm}}$. 

For the other direction, assume that $\mathsf{Enl}_\bot((R_\bot)_\div)$ satisfies $E_{\mathsf{ftc-cm}}$ and suppose $a \in R_\bot$ satisfies $0 \cdot a \neq 0$ and $a \neq \bot$. We consider $b = 1 + 0 \cdot a$. We first notice $b \neq \bot$, as otherwise: $\bot = \bot + (-1) =  b + (-1) = (1 + 0 \cdot a)   + (-1) = 0 \cdot a$ so that $a = (1+0) \cdot a = a + \bot = \bot$. Thus if $b^{-1}$ is defined in $\mathsf{Enl}_\bot((R_\bot)_\div)$ it must be a  proper inverse of $b$.
We find $ b^{-1} = (1 + 0 \cdot a)^{-1} = 1 + 0 \cdot a = b$, so that 
$$ 1 = b \cdot b^{-1} = (1 +  0 \cdot a) \cdot (1 +  0 \cdot a) = 1 + (0 + 0)  \cdot a + (0 \cdot a) \cdot (0 \cdot a) = 1 + 0 \cdot a$$ 
thus $0 \cdot a = 0$ in contradiction with the assumptions on $a$.
\end{proof} 

We notice that a WCR which does not satisfy $0 \cdot x \neq 0 \to x = \bot$ is easy to find: add a constant $c$ and take an initial algebra of $K = \mathsf{Alg}(\Sigma_{r, \bot, c}, E_{\mathsf{ftc-cm},c})$.

In~\cite{DiasD2023} one finds results concerning for precisely which WCR's have an expansion to a model of $E_{\mathsf{ftc-cm}}$ (i.e., to a generalised common meadow).

Recalling from Section \ref{PAs} the notation of~\cite{BergstraT2022TM}, for a common meadow $M$, $\mathsf{Pdt}_\bot(M)$ is a meadow (= field with partial division). For a commutative ring with common division $R_{\div, \bot}$, $\mathsf{Pdt}_\bot(R_{\div, \bot})$ is a commutative ring with partial division, and for an integral domain with common division $R_{\div, \bot}$, 
$\mathsf{Pdt}_\bot(M)$ is an integral domain with partial division.



\section{Examples of conditional equations}\label{conditional_equations}


\subsection{Deriving conditional equations}
We will comment on the logic of common meadows which seems to allow for a rather tailor made development using an additional auxiliary operator $\zeta$ and various dedicated proof rules in combination with the rule $\mathsf{R_{cm}}$ mentioned above.

Obviously, all non-trivial ($0 \neq 1$) commutative rings with inverse based common division satisfy the conditional equation
$0 = 1 \to 0 = \bot.$ More generally:

\begin{proposition}  
$E_{\mathsf{ftc-cm}} \vdash 0 = 1 \to 0 = \bot$.
\end{proposition}
\begin{proof}$E_{\mathsf{ftc-cm}} \cup \{0=1 \} \vdash
 0 =1 = \frac{1}{1} = \frac{1}{0} = \bot$.
\end{proof}

However, the structure $A_0$ in which $0=1, \bot \neq 0$ and
$$ 0 +0 = 0+1 = 0 \cdot 0 = 0 \cdot 1 =0, \  \bot + 0 = \bot + 1 = \bot + \bot= \bot$$
satisfies $E_{\mathsf{wcr},\bot}$ while it rejects $0 = 1 \to 0 = \bot$.
We find that $E_{\mathsf{wcr},\bot} \not\vdash 0 = 1 \to 0 = \bot.$ It follows that $E_{\mathsf{ftc-cm}}$ is 
not a conservative extension of $E_{\mathsf{wcr},\bot}$.

As a strengthening of $E_{\mathsf{ftc-cm}} \vdash 0 = 1 \to 0 = \bot$ we find:
\begin{proposition}
\label{aux1}
$E_{\mathsf{ftc-cm}} \vdash x = x+ 1 \to x = \bot.$
\end{proposition}
\begin{proof} We will make implicit use of the axioms of $E_{\mathsf{ftc-cm}}$. Suppose $x = x + 1$ then $x -x = (x+1) - x= (x-x) + 1$ so that
$0 \cdot x = 0 \cdot x + 1$. 
Now $0 \cdot x = 1 + 0 \cdot x =  \frac{1} {1 + 0 \cdot x} = \frac{1} {0 \cdot x} =   \frac{1}{0} \cdot \frac{1}{ x}= \bot \cdot \frac{1}{ x} = \bot$. Thus $x = 1 \cdot x = (1 + 0) \cdot x = x + 0 \cdot x =  x + \bot = \bot$.
\end{proof}

Similarly we find:
\begin{proposition} 
$E_{\mathsf{ftc-cm}} \vdash 0 \cdot x =  \frac{1}{y}  \to x = \bot.$
\end{proposition}
\begin{proof} 
$E_{\mathsf{ftc-cm}}$ implies $\frac{y}{y} = 1 + \frac{0}{y}$ so that with $0 \cdot x =  \frac{1}{y}$   we find 
$\frac{y}{y} = y \cdot \frac{1}{y} = y \cdot 0 \cdot x = 0\cdot x + 0 \cdot y = 1 +  \frac{0}{y} = 1 + 0 \cdot \frac{1}{y} = 
1 + 0 \cdot 0 \cdot x = 1 + 0 \cdot x$ Upon adding $0 \cdot y$ on both sides  
$0\cdot x + 0 \cdot y + 0\cdot y = 1 +  0\cdot x + 0\cdot y \ \textrm{and} \ 0\cdot x + 0 \cdot y = 1 +  0\cdot x + 0\cdot y$ so that 
$0\cdot x + 0 \cdot y  = 1 +  0\cdot x + 0\cdot y $
 and with Proposition \ref{aux1} $0\cdot x + 0 \cdot y = \bot$. Then
with  $0\cdot x + 0 \cdot y = 1 +  0\cdot x$  as derived above  we have $1 +  0\cdot x= \bot $ from which $x= \bot$ easily follows. \end{proof}

\begin{proposition} If $M \models x + y = \bot \to t=r$ in all common meadows $M$ then 
$E_{\mathsf{ftc-cm}} \vdash x + y = \bot \to t=r$.
\begin{proof} 
Suppose $t$ contains both $x$ and $y$ as free variables then, 
using arguments from~\cite{BergstraP2015} or~\cite{BergstraT2023arxiv}, 
$E_{\mathsf{ftc-cm}} \vdash t = t + 0\cdot x + 0 \cdot y = \bot$, 
so that in all common meadows $r = \bot$. 
In this case apply fracterm flattening and choose terms $p$ over the signature of rings such that 
$E_{\mathsf{ftc-cm}} \vdash r = \frac{p}{q}$. 
It follows that in all common meadows (and for that reason in all fields $F$), 
$q = 0$ so that $E_{\mathsf{ftc-cm}} \vdash q = 0$ and $E_{\mathsf{ftc-cm}} \vdash r = \bot$, 
whence  $E_{\mathsf{ftc-cm}} \vdash t = r$. 

We may then focus on the case that
both $t$ and $r$ contain at most one variable of $x$ and $y$. We assume that $t$ does not contain $y$ 
and then we distinguish two cases: 

(i) $x$ occurs in $t$ and $y$ occurs in $r$, in which case 
$E_{\mathsf{ftc-cm}} \vdash t = t + 0\cdot x$ and $E_{\mathsf{ftc-cm}} \vdash r = r + 0\cdot y$ so that in all common meadows
$\bot = t + 0 \cdot x + 0 \cdot y = t  + 0 \cdot y = r + 0 \cdot y = r$. From this observation we find that, upon choosing $x= \bot$,
$\bot = r$ holds in all common meadows so that  it is derivable from 
$E_{\mathsf{ftc-cm}} $ according to the completeness result of~\cite{BergstraT2023arxiv}.
The case that $y$ occurs in $t$ and $x$ occurs in $r$ is dealt with symmetrically.

(ii) $y$ does not occur in $t$ in which case 
 $y $ occurs in neither $t$ and $r$, 
then choosing $y = \bot$ we find that $t = r$ is valid in all common meadows so that it is derivable from 
$E_{\mathsf{ftc-cm}} $ according to the completeness result of~\cite{BergstraT2023arxiv}. 
\end{proof}
\end{proposition}


\subsection{AVL and conditional equations}\label{completeness_conjecture}

As an example of the use of this rule consider the following conditional equation:
$$\phi_{2,3,5} \equiv \frac{0}{2} \cdot  \frac{0}{3} =  \frac{0}{2} \cdot  \frac{0}{5} \to  \frac{0}{3} =  \frac{0}{5}. $$

\begin{lemma}
$\phi_{2,3,5}$ holds in all common meadows.
\end{lemma}

\begin{proof}
This is a simple case distinction. Note that the conclusion may be written as $0 \cdot \frac{1}{3} = 0 \cdot  \frac{1}{5}$.

Assuming $\frac{0}{2} \cdot  \frac{0}{3} =  \frac{0}{2} \cdot  \frac{0}{5} $ and $\frac{1}{3} = \bot$ we will show that 
$\frac{0}{5} =  \bot $, leaving the proof that with $\frac{1}{5} =  \bot$ one may derive $\frac{1}{3} = \bot$ to the reader. 
 If $\frac{1}{3} = \bot$ then it is immedate that $ 0 \cdot 3 = 0 \cdot (1 + 1 + 1) = 0$, moreover with $\mathsf{AVL}$ one finds $0 \cdot 3 = 3$ so that $3 = 0$. Then $ 5 \cdot 2 = 1 + 9 = 1\!\!\mod 3$, so that $2$ is the inverse of $5$ and conversely. It follows that 
 $0 \cdot (5 \cdot 2) = 0 \cdot 1 = 0$ and thus $0 \cdot 5 = 0$ and $0 \cdot 2 = 0$ so that $\frac{0}{5} =  0$ and $\frac{0}{2} = 0$.
 
Now the condition   $ \frac{0}{2} \cdot \frac{0}{3} =  \frac{0}{2} \cdot  \frac{0}{5}$ can be used, and upon substitution of the now known facts this yields: $0 \cdot \bot = 0 \cdot 0$ and thus $\bot = 0$, which suffices to obtain $\frac{0}{3} =  \frac{0}{5}$.
\end{proof}

\begin{proposition} 
$E_{\mathsf{ftc-cm}}+ \mathsf{AVL} \not \vdash \phi_{2,3,5} $.
\end{proposition}

\begin{proof} A model of $E_{\mathsf{ftc-cm}}+ \mathsf{AVL}$ in which $\phi_{2,3,5} $ fails can be found following the structure theory for GCM's as set out in~\cite{DiasD2023,DiasD2024}. The underlying idea is that all GCM's $M$, contain a lattice $L_M$ made up by elements from $0 \cdot M$. Let $m \in 0 \cdot M$ then $m$ determines the subring of $M$ which consists of 
multiples of $1 + m$. 

A common meadow always has 2-element lattice containing $0$ and $\bot$ only. 
The same holds for commutative rings with inverse based common division. 
Each element $b \in L_M$ of the lattice serves as the zero for a ring $R^b$. 

We consider a structure with a three element lattice containing: $0, c, \bot$.  
Here  $c$ can be understood as an element which satisfies $c = 0 \cdot c$ 
and $c = c +  \frac{0}{p} $ for all primes $p$. $R_{0} $ is a copy of $\Int(\frac{1}{3})$, 
that is the ring of integers augmented with a constant for $\frac{1}{3}$ 
which satisfies $3 \cdot \frac{1}{3} = 1$ and $R^c$, 
the substructure generated by $c$ is a common meadow of 
rational numbers,  denoted as $Q_{\bot}$. 

We assume that the domain of $R_{0} $ is made up of simplified fracterms with as denominators powers of $3$ only, while 
the domain of $R_{0} $ is made of of all simplified fracterms. The natural isomorphism $\iota \colon 
\Int(\frac{1}{3}) \to Q_{\bot}$ (i.e., from $R_{0} $ to $R_{c} $) is needed to describe the operations of $M$:

$p \star q = p \star_{R_c} q$,  for all pairs  $p,q \in R_c$, $\star \in \{+,-, \cdot ,\div\}$;

$p \star q = \iota(p) \star_{R_c} q$,  for all pairs  $p \in R_0,q \in R_c$, $\star \in \{+,-, \cdot ,\div\}$;

$p \star q = p \star_{R_c} \iota(q)$,  for all pairs  $p \in R_c,q \in R_0$, $\star \in \{+,-, \cdot ,\div\}$;

$p \star q = p \star_{R_0}q$,  for $p,q \in R_0$, $\star \in \{+,-, \cdot \}$;

$p \div 0 = \bot$ for $p \in R_0$;

$p \div q = \iota (p) \div \iota(q)$ for all $q$ with a numerator that is not a power of $3$.

\noindent It is an elementary check of al axioms that $M$ satisfies $E_{\mathsf{ftc-cm}}+ \mathsf{AVL}$ and indeed 
$M\not \models  \phi_{2,3,5}$.
\end{proof} 

We expect that the conditional equational theory of common meadows has no finite axiomatisation, more specifically we formulate the following conjecture:
\begin{conjecture} 
There is no finite set of conditional equations $E$ true of all 
common meadows such that for all prime numbers $p \geq 5$, $E \vdash   \frac{0}{2} \cdot  \frac{0}{3} =  \frac{0}{2} \cdot  \frac{0}{p} \to  \frac{0}{3} =  \frac{0}{p} $.
\end{conjecture}


\section{Customising the logic of common meadows and completeness}\label{customising_logic}

\subsection{Adding a new proof rule}

The following proof rule is sound for common meadows, with $E$ a conjunction of equations: 
$$\mathsf{R_{cm}}\,\,\frac{\vdash E \wedge  t = \bot  \to r =\bot, \  \vdash E \wedge r = \bot  \to t=\bot}
{ \vdash E \to 0 \cdot t=0 \cdot r}$$

We denote derivability in the proof system of conditional equational logic extended with 
$\mathsf{R_{cm}}$ by $\vdash_{\mathsf{R_{cm}}}$. 
The following observation is immediate.
\begin{proposition} 
\label{Rcm-sound1}
The proof system $\vdash_\mathsf{R_{cm}}$ is sound for commutative rings with inverse based common division.
\end{proposition}

As $\mathsf{AVL}$ fails in commutative rings  with inverse based common division that are not fields (Proposition \ref{axiom _meadows_relative_to_rings}), and in view of Proposition~\ref{Rcm-sound1}, $\mathsf{AVL}$ is not derivable with $\vdash_{\mathsf{R_{cm}}}$ from $E_{\mathsf{ftc-cm}}$.

We will write $\models_{\mathsf{cm}}$ for satisfaction in all common meadows and $\vdash_{\mathsf{cm/avl}}$ for derivability from $E_{\mathsf{ftc-cm}}+ \mathsf{AVL}$.

\begin{lemma}\label {completenessA} 
If $\models_{\mathsf{cm}} E\to \frac{1}{t} = \bot$ then $\vdash_{\mathsf{cm/avl}} E \to\frac{1}{t} = \bot$.
\end{lemma}

\begin{proof} 
Suppose $\not\vdash_{\mathsf{cm/avl}} \frac{1}{t} = \bot$ then for some GCM $M$ which satisfies $
\mathsf{AVL}$ it is the case that for some substitution $\vec{c}$ of the 
free variables of $E \to \frac{1}{t} = \bot$ it is the case that $\frac{1}{t(\vec{c})} \neq \bot$. We have to find a common meadow in which this is also true.

Consider the mapping $\phi(x) = x + 0 \cdot \frac{1}{t}$. 
First of all, we notice that $\phi(1) \neq \bot$, as otherwise it would be so that 
$1 + 0 \cdot \frac{1}{t} = \bot$ which implies $\frac{1}{t} = \bot$. 
Thus, $\phi$ is a homomorphism to a GCM, say $M'$, where 
$$t(\vec{\phi(c)})\cdot \frac{1}{t(\vec{\phi(c)})} = 1 + \frac{0}{t(\vec{\phi(c)})} = \phi(1 + \frac{0}{t(\vec{c})}) = \phi(1) = 1$$ 
which also satisfies 
$\mathsf{AVL}$ in view of Lemma~\ref{HomToGCM1}. 
In $M'$, $E \to \frac{1}{t} = \bot$ is not valid on the substitution $\vec{\phi(c)}$. 
To see the this, notice that the homomorphism $\phi$ preserves all conditions while 
$\frac{1}{t(\vec{\phi(c)})}$ implies $ 1 = \bot $ against the assumption on $M$. 
With Lemma~\ref{HomToCM2} we find the existence of a homomorphic $M''$  image of $M'$ which is a common meadow and in which the image of $t(\vec{\phi(c)}$ is invertible so
that $\frac{1}{t(\vec{c})} \neq \bot$.
\end{proof}

\begin{lemma}\label{completenessB}
If $\models_{\mathsf{cm}} E\to t = 0 \cdot t$ then $\vdash_{\mathsf{cm/avl}} t = 0 \cdot t$.
\end{lemma}

\begin{proof} Assuming $\models_{\mathsf{cm}} E\to t = 0 \cdot t$ we have $\models_{\mathsf{cm}} E\to \frac{1}{t} = \bot$
so that with Lemma~\ref{completenessA} $\vdash_{\mathsf{cm/avl}} E \to\frac{1}{t} = \bot$ whence with $\mathsf{AVL}$
$\vdash_{\mathsf{cm/avl}} E \to t = 0 \cdot t = \bot$.
\end{proof}

\begin{theorem} 
If $\models_{\mathsf{cm}} E\to t = \bot$ then $\vdash_{\mathsf{cm/avl}} E \to t = \bot$.
\label{completenessB2}
\end{theorem}
\begin{proof} Using fracterm flattening we may write 
$t$ as $\frac{p(\vec{x})}{q(\vec{x})}$ where $p$ and $q$ are division free 
and, moreover, it may be assumed that all variables occurring in $p$ are also present in $q$. So, we assume that 
$\models_{\mathsf{cm}} t = \frac{p(\vec{x})}{q(\vec{x})}$. 
Assuming $\models_{\mathsf{cm}} E \to \frac{p(\vec{x})}{q(\vec{x})} = \bot$ we have 
that in each common meadow for each substitution either the numerator of $\frac{p(\vec{x})}{q(\vec{x})}$ equals $\bot$  
or the denominator equals $0$ or $\bot$. Now if the numerator is $\bot$ then for at least one variable $\bot$ has been substituted so that the denominator equals $\bot$ as well. In each of these cases $q(\vec{x})\cdot 0 = q(\vec{x})$ in the substitution at hand. It follows that 
$\models_{\mathsf{cm}} E \to q(\vec{x})\cdot 0 = q(\vec{x})$ so that, with Lemma~\ref{completenessB},
$\vdash_{\mathsf{cm/avl}} q(\vec{x})\cdot 0 = q(\vec{x})$ from which it follows trivially that 
$\vdash_{\mathsf{cm/avl}} t = \bot$
\end{proof}

\begin{lemma}
\label{completenessC} If 
$\models_{\mathsf{cm}} E \wedge 0 \cdot t = 0 \cdot r \to t = r$ then 
$\vdash_{\mathsf{cm/avl}} E \wedge 0 \cdot t = 0 \cdot r \to t = r$.
\end{lemma}
\begin{proof} In the style of the proof of Lemma~\ref{completenessA} 
 we consider a GCM $M$ that satisfies $\mathsf{ASL}$ and which refutes $E \wedge 0 \cdot t = 0 \cdot r \to t = r$.
 For some substitution $\vec{c}$ for the relevant free variables the conditions are true while $t$ and $r$ take different values: 
  in $M$: $t(\vec{c}) \neq r(\vec{c})$. Consider $s = \frac{1}{t(\vec{c}) - r(\vec{c})}$ in $M'$. We claim that $s \neq \bot$. Otherwise, with $\mathsf{AVL}$, $t(\vec{c}) - r(\vec{c}) = 0 \cdot (t(\vec{c}) - r(\vec{c}))$, so that 
  $t(\vec{c}) - r(\vec{c}) + r(\vec{c})= 0 \cdot (t(\vec{c}) - r(\vec{c})) + r(\vec{c})$ and, with $x + 0\cdot x = x$ and the condition 
  $0 \cdot t(\vec{c}) = 0 \cdot r(\vec{c})$ it follows that $t(\vec{c}) = r(\vec{c})$ in contradiction with the assumptions on $M'$. 
  Now we may consider the homomorphism $\phi(x) = x + 0 \cdot s$. 
  The image of $\phi$ is a GCM, that satisfies $\mathsf{AVL}$ in which 
  $t(\vec{c}) \neq r(\vec{c})$ differs from $\bot$, as otherwise $1 = t(\vec{c}) \neq r(\vec{c}) \cdot s= \bot$. Thus $M'$ refutes
  $E \wedge 0 \cdot t = 0 \cdot r \to t = r$, but now on a substitution $\vec{\phi(c)}$ which makes $t-r$ invertible. 
  Using Theorem~\ref{HomToCM2} we find a homomorphic image $M'' = \psi(M')$ of $M'$ which is a common meadow and in which
  $E \wedge 0 \cdot t = 0 \cdot r \to t = r$ is refuted on the substitution $\psi(\phi(\vec{c}))$.  
\end{proof}

Completeness can be obtained with the help of the additional rule $\mathsf{R_{cm}}$:

\begin{theorem}  
$\models_{\mathsf{cm}} E  \to t = r$ if, and only if, 
$E_{\mathsf{ftc-cm}}+ \mathsf{AVL} \vdash_{\mathsf{cm}} E  \to t = r$.
\end{theorem}

\begin{proof}
Soundness of the proof system for its purpose is obvious. For completeness, we notice that if 
$\models_{\mathsf{cm}} E  \to t = r$ also $\models_{\mathsf{cm}} E \wedge t = \bot   \to r = \bot $ with Lemma~\ref{completenessA} also $\vdash_{\mathsf{cm/avl}} E \wedge t = \bot   \to r = \bot $. Similarly, symmetrically,  
$\vdash_{\mathsf{cm/avl}} E \wedge r = \bot   \to t = \bot $.  So, with $\mathsf{R_{cm}}$ we obtain
$E_{\mathsf{ftc-cm}}+ \mathsf{AVL} \vdash_{\mathsf{cm}}  E \to 0 \cdot t=0 \cdot r$. 

Obviously, $\models_{\mathsf{cm}}  E \wedge 0 \cdot t = 0 \cdot r \to 0 \cdot t=0 \cdot r$, as the additional condition does no harm.  With Lemma~\ref{completenessC}, we find
$\vdash_{\mathsf{cm/avl}} \vdash  E \wedge 0 \cdot t=0 \cdot r \to t=r$, and in combination with  
$E_{\mathsf{ftc-cm}}+ \mathsf{AVL} \vdash_{\mathsf{cm}}  E \to 0 \cdot t=0 \cdot r$ we find
$E_{\mathsf{ftc-cm}}+ \mathsf{AVL} \vdash_{\mathsf{cm}}  E \to t = r$.
\end{proof}


We do  not know whether or not the use of an additional rule --  $\mathsf{R_{cm}}$ or otherwise -- is necessary for a finite and complete axiomatisation of the conditional equational theory of common meadows, although as stated above we expect this 
 to be the case. Nevertheless, a foundational open problem regarding common meadows remains Problem~\ref{MainProblem}.

Instead of the rule $\mathsf{R_{cm}}$ we may also use a less symmetric, and less intuitive, but technically simpler rule $\mathsf{R_{cm}'}$
 $$\mathsf{R_{cm}'}\,\,\frac{\vdash E \wedge  t = \bot  \to r =\bot}
{ \vdash E \to 0 \cdot r=0 \cdot (t+r)}$$

Using $\mathsf{R_{cm}'}$ it is easy to obtain $\mathsf{R_{cm}}$ as a derived rule. The other direction is not obvious, however, and we leave open the question whether or not $\mathsf{R_{cm}'}$ is stronger than $\mathsf{R_{cm}}$.

\subsection{Expanding common meadows with a conditional operator}

Finally, we extend the signature of common meadows with a conditional operator. 
In particular, we propose to work with  
$x \vartriangleleft y  \vartriangleright z$ 
understood as 
\begin{center}
``if $y=0$ then $z$ else $x$, where $\bot$ is returned if $y = \bot$". 
\end{center}
The conditional operator is not strict in all arguments, e.g., $0 \vartriangleleft 1  \vartriangleright \bot = 0$. 
Nevertheless, it is easily shown from the axioms in 
Table~\ref{CondOp}, that for $a \neq \bot$, $\phi_a(x) = x + 0 \cdot a$ is a 
homomorphism w.r.t. the conditional operator. It follows that the completeness result of Theorem~\ref{completenessB} can be 
generalised to the signature of common meadows expanded with the conditional operator.

\begin{table}
\centering
\hrule
\begin{align}
{\tt import~}  &~ E_{\mathsf{wcr},\bot} + \
\mathsf{AVL} \nonumber \\
x \vartriangleleft 0\vartriangleright y &=  y\\
x \vartriangleleft \bot \vartriangleright y &=  \bot\\
y \cdot u = 1 \to x \vartriangleleft y \vartriangleright z &=  z\\
x \vartriangleleft (y + (0 \cdot u))\vartriangleright z &=  (x \vartriangleleft y  \vartriangleright z) + (0 \cdot u) \\
(x + (0 \cdot u)) \vartriangleleft y \vartriangleright (z + (0 \cdot u)) &=  (x \vartriangleleft y  \vartriangleright z) + (0 \cdot u) 
\end{align}
\hrule
\medskip
\caption{$E_{\mathsf{Cond-Op}}$: axioms for the conditional operator}
\label{CondOp}
\end{table}


\section{Integral domains with $\bot$}\label{integral_domains}

\subsection{Basics}
A {\em integral domain with common division} is a commutative ring with common division $\mathsf{Enl}_\bot(R_{\div})$ that is based on a ring $R$ that is an integral domain, i.e., it has no zero divisors: for all $a, b \in R$, 
$$a \cdot b = 0 \to a = 0 \vee b = 0.$$ 
Let $\mathsf{CID}_{\div, \bot}$ be the class of all such algebras. We notice that in personal communication Bruno Dinis suggested a somewhat stronger requirement for integral domains with inverse based common division, the implications of which we have not investigated. 
$$a \cdot b = 0 \cdot a \cdot b \to 0 \cdot a = a \vee 0 \cdot b  = b.$$ 

Clearly, every common meadow is an integral domain with common division, and each integral domain with common division is a commutative ring with common division, and both inclusions are proper. In symbols,
$$\mathsf{CM} \subsetneqq \mathsf{CID}_{\div, \bot} \subsetneqq \mathsf{CR}_{\div, \bot}.$$

\begin{proposition}
Every integral domain with common division can be extended to a common meadow.
\end{proposition}

\begin{proof} 
Consider an integral domain with common division $M= \mathsf{Enl}_\bot(R_{\div})$ where the ring $R$ is an integral domain. Let $R'$ be a field of fractions for $R$ then $M' = \mathsf{Enl}_\bot((R')_{\div})$ is a common meadow which extends $M$.
\end{proof} 

Because integral domains with common division are a subclass of commutative rings with common division, the conditional equational theory of integral domains with common division includes the conditional equational theory of commutative rings with common division. 

\begin{proposition}
The conditional equational theory of integral domains with common division is a proper extension of the conditional equational theory of commutative rings with common division. In fact,
these inclusions are proper: 
$$ConEqnThy(\mathsf{CR}_{\bot,\div}) \subsetneqq ConEqnThy(\mathsf{CID}_{\bot,\div}) \subsetneqq ConEqnThy(\mathsf{CM}) $$ 
\end{proposition}

\begin{proof} 
Fir the first proper extension we notice that $\mathsf{AVL}$ is true in all common meadows but not in the ring of integers enlarged with inverse based common division.
For the second proper inclusion we notice that  the conditional equation 
$\phi \equiv x \cdot x = 0 \to x = 0$ is true in any integral domain with common division, while it fails in a commutative ring with common division based on a ring of dual numbers (we refer to~\cite{Bergstra2019c} for dual numbers in a setting of meadows).
\end{proof}

The conditional equation $\phi \equiv x \cdot x \cdot y = 0 \to x \cdot y = 0$ is true for all integral domains, it is provable from
$E_{\mathsf{ftc-cm}} + \mathsf{AVL}$ but it is stronger than $x \cdot x = 0 \to x = 0$ as it excludes all rings with non-zero nil-potent elements.

Again we have: 

\begin{proposition}
The class of integral domains with common division  is not a quasivariety. 
\end{proposition}

\begin{proof} 
Suppose that conditional equational theory $E_{ce}$ has precisely the commutative rings with common division as it models.
We introduce fresh constants $c$ and $d$. We consider an initial algebra of $E' = E_{ce} \cup \{ c\cdot  d = 0\}$. 
Now, there is a 
ring with common division $M_1$ that satisfies $E'$ by taking $c= 1$ and $d = 0$ from which it follows  that 
$E' \not \vdash c = 0$.
Similarly, there is a model $M_2$ of $E'$ in which $d=0$ and $c = 1$  so that $E' \not \vdash d= 0$. 
It follows that in an initial algebra of $E'$, $c \cdot d = 0$ while $c \neq 0$ and $d\neq 0$. 

However, the latter situation cannot occur in any structure of the form 
$\mathsf{Enl}_\bot(R_{\div})$ with $R$ an integral domain, because $c \cdot d = 0$ implies 
$c\neq \bot$ and $d \neq \bot$ so that  $c \in R$ and $d \in R$, 
whence $c\cdot d = 0 \in R$, so that $c\cdot d = 0 $ in $\mathsf{Enl}_\bot(R_{\div})$
\end{proof}

While the common meadows can be characterized with a single conditional equation as a subclass of the 
commutative rings with common division (Proposition \ref{axiom _meadows_relative_to_rings}), a corresponding characterization of the  integral domains with common division  is not so easy to find:

\begin{problem} 
Is there a set of conditional equations $T^{ce}_{idcd} $ true of all  integral domains with inverse based common division  such that a commutative ring with common division is an  integral domain with common division  if, and only, if it satisfies $T^{ce}_{idcd}$? If so is there a finite set 
$T^{ce}_{idcd}$  with the mentioned properties?
\end{problem}
The equational theory of integral domains with inverse based common division is the same as the equational theory of common meadows and as the equational theory of commutative rings with inverse based common division.
The following issue is left open:

\begin{problem} Is there a finite axiomatization with conditional equations of the conditional equational theory of integral domains with inverse based common division?
\end{problem}


\subsection{A second definition of division for integral domains}

Recall our convention where we change notation from $\div$ to the familiar fraction notations when appropriate. We now introduce a second kind of division (for integral domains) and extend our flexible convention relying on the context. Thus, as long as the divisions are declared, fracterms can continue have these forms 
$$p \div q, \ \displaystyle \frac{p}{q}, \ p/q.$$

Here is another definition for division in case the ring $R$ is an integral domain:

\begin{definition}\label{division_is_well-defined}
Let $R$ be an integral domain. Then it has a unique {\em partial division operator} $-/-$ and defined by 

(i) if $b \neq 0$ and $ b\cdot c = a$ ($c$ is unique in this respect) then $ a / b =_{\mathit{def}} c$,

(ii) if $b= 0$ or if for no $c \in R$, $ b \cdot c = a$ then $a/ b $ is undefined.
\end{definition}

\begin{lemma}
Let $R$ be an integral domain. In case (i) of Definition \ref{division_is_well-defined} above if such $c$ exists then it is unique.
\end{lemma}

\begin{proof}
If such $c$ exists it is unique because of the absence of zero-divisors in $R$: if also $ b\cdot c' = a$ then $a - a = b\cdot c - b\cdot c' = b\cdot (c - c' ) = 0$, and then, given that $b \neq 0$, $c-c'=0$, and thus $c = c'$.
\end{proof}

Extending a ring $R$ with this operator $/$ creates a {\em ring with  partial general division} denoted $R_{/}$.  This general form of division is more familiar in arithmetic.

\begin{example}
{\em In $\Int$ there are more divisions (in comparison with inverse based division): the equation $b \cdot x = 1$ (for defining $\frac{1}{b}$) has solutions whenever $x$ divides $b$, e.g. $12/1 =12, 12/2 = 6, 12/3 = 4, 12/4 = 3$, and $ 12/6 = 2$ all exist and are unique.

The requirement that we work over a ring without zero divisors is needed. For instance in the finite ring  $\Int_{6}$, the equation $4 \cdot x = 2 \ mod \ 6$ has  solutions $2/4 = 2$ because  $2 = 2 \cdot 4 \ mod \ 6 = 8 \ mod \ 6 =2$ and  $2/4 = 5$ because  $2 = 5 \cdot 4 \ mod \ 6 = 20 \ mod \ 6 =2$. So, unless working over an integral domain, divisions are not necessarily unique but may be multivalued. This complication is absent for inverse based division, however.}
\qed
\end{example}

However, $\mathsf{Enl}_\bot(R_{/})$ does not satisfy the equations of $E_{\mathsf{ftc-cm}}$:

\begin{proposition}\label{general_divison_fails_axioms}
$\mathsf{Enl}_\bot(R_{/})\not \models E_{\mathsf{ftc-cm}}$.
\end{proposition}
\begin{proof}
For instance with $R=\Int$,
$\mathsf{Enl}_\bot(R_{/}) \models \frac{1}{3} \cdot \frac{3}{1} = \bot \cdot 3 = \bot$ while 
$\mathsf{Enl}_\bot(R_{/}) \models \frac{1 \cdot 3}{3 \cdot 1} = \frac{3}{3} = 1$, which refutes equation \ref{fracDivMult} of $E_{\mathsf{ftc-cm}}$.
\end{proof}
Actually, $\mathsf{Enl}_\bot(\Int_{/}) $ satisfies the axioms of 
 $E_{\mathsf{ftc-cm}}$ with the important exceptions of equations~\ref{fracDivMult}, and~\ref{fracDivSum}.
Given Proposition \ref{general_divison_fails_axioms}, a typical open problem regarding equational logic over these rings with common division is as follows:

\begin{problem}
\label{FTFZ}
 Does  $\mathsf{Enl}_\bot(\Int_{/})$ allow fracterm flattening? 
 \end{problem}

As far as we can see, there is no straightforward approach to equational reasoning in $\mathsf{Enl}_\bot(R_{/})$ with standard first order (i.e., Tarski) semantics. 
Yet, we consider $\mathsf{Enl}_\bot(R_{/})$ to be a more natural structure than $\mathsf{Enl}_\bot(R_{\div})$ because direct common division is more often defined (infinitely often in the case of $R=\Int$). Defining division with the use of the multiplicative inverse as an intermediate step is fine when working in a field, but doing so is less convincing when working in a ring. However, clearly, the equations  for $\mathsf{Enl}_\bot(R_{/})$ are  less 
attractive than the equations for $\mathsf{Enl}_\bot(R_{\div})$.

\section{Rings with division and eager equality semantics}\label{Rings_with _division _and_eager_equality}


\subsection{Eager equality}\label{eager_equality}

In \cite{BergstraT2022_ToCL,BergstraT2023_CJ} we have introduced and investigated in some detail an alternate way of identifying two expressions with partial operations. Called 
\textit{eager equality}, and denoted $\approx$,  its general form is as follows. 

Let $\llbracket t \rrbracket$ denote the value of term $t$ in a partial algebra $A$ and write  $t \downarrow$ if $\llbracket t \rrbracket$ is defined
and $t \uparrow$ if $\llbracket t \rrbracket$ is not defined. Naturally, when the expressions are defined in $A$, eager equality of expressions is the same as equality $=$ in $A$:
$$\textrm{if}  \ t  \ \downarrow \textrm{and} \  t' \downarrow \textrm{then:} \ \llbracket t \rrbracket = \llbracket t' \rrbracket \iff t \approx t'.$$ 

However, if one or other, or both, of $t$ and $t'$ are not defined then eager equality $\approx$ satisfies
$$ t \uparrow \textrm{or/and} \ t' \uparrow  \implies  t \approx t'.$$
Thus, eager equality means that if an expression is undefined then \textit{all other other expressions} -- defined or not --  will be deemed equal to it. This motivates the description `eager'.
Finally, we note that
$$t \not\approx t' \iff  t  \ \downarrow \textrm{and} \  t' \downarrow  \textrm{and} \  t \neq t'.$$

Eager equality is symmetric and reflexive, but it is \textit{not} transitive. Eager equality satisfies what we call  safe transitivity: 
$$x \approx y \wedge y \approx z \wedge y \not\approx u \to x \approx z.$$
It does satisfy the congruence property, i.e., operations of $A$ preserve $\approx$.  Eager equality contrasts with one of the oldest notions of partial equality, that of \textit{Kleene equality}, which requires both the expressions to be either defined and equal, or both undefined \cite{Kleene1952}. As emphasised in our \cite{BergstraT2022_ToCL}, eager equality is something rather different and unusual.

Let $A$ be be a total or partial algebra and let non-empty $V \subseteq A$, then with $\langle A \rangle_V$ we denote the subalgebra of $A$ with domain limited to $V$ and operations made partial on inputs in $V$ when their values lie outside $V$  or do not exist. We compare the totalisations $\mathsf{Enl}_\bot(A)$ with $\mathsf{Enl}_\bot(\langle A \rangle_V)$. The proof of the following observation is straightforward:

\begin{proposition}\label{eager_transfer}
Given a partial algebra $A$ of signature $\Sigma$ and an equation $t=r$ over $\Sigma_\bot$,  then $\mathsf{Enl}_\bot(A)\models t=r$ implies  $\mathsf{Enl}_\bot(\langle A \rangle_V)\models_{\mathit{eager}} t=r$.
\end{proposition}

\begin{proof}
The modification that takes place by turning $\mathsf{Enl}_\bot(A)$ into $\mathsf{Enl}_\bot(\langle A \rangle_V)$ only introduces more occurrences of $\bot$ as a function value. Such modifications cannot invalidate an equation w.r.t. to eager equality.
\end{proof}

\begin{example}
{\em  We are interested in instances of this construction when $A = F$ is a field and $V = R \subseteq F$ is a ring. Then, with $F_\div$ denoting $F$ enriched with a partial division operator, we find that $\mathsf{Enl}_\bot(\langle F_\div \rangle_{R})$ will satisfy in eager equality semantics all equations -- including and especially  $E_{\mathsf{ftc-cm}}$ -- which  $F_\div$ satisfies in standard first order semantics. A typical example of this state of affairs is found with $F = \rat$ and $R = \Int \subseteq \rat$. }
\qed
\end{example}

\begin{example}
{\em Another instance of the construction is where \textit{bounded} common meadows are obtained from a common meadow, as discussed in~\cite{BergstraT2022c}.  This works as follows:  We restrict the domain of $\mathsf{Enl}_\bot(\rat_\div)$  to an interval $V = (-b,b)$ for some rational bound $b>1$ to make
$\mathsf{Enl}_\bot(\langle \rat_\div \rangle_V)$. We find that bounded meadows thus defined fail to satisfy the equations of $E_{\mathsf{ftc-cm}}$ in ordinary first order semantics. However, the equations of  $E_{\mathsf{ftc-cm}}$
are satisfied in $\mathsf{Enl}_\bot(\langle \rat_\div \rangle_V)$ in eager equality semantics.}
\qed
\end{example}


\subsection{Eager equality and the axioms for common for
division}\label{eager_equality_common division}

Now we notice that for an integral domain $R$, $\mathsf{Enl}_\bot(R_{/}) $ eager equality   validates all of $E_{\mathsf{ftc-cm}}$ and allows fracterm flattening for that reason. 

\begin{proposition} 
Let $R$ be an integral domain. Then 
$\mathsf{Enl}_\bot(R_{/}) \models_{eager} E_{\mathsf{ftc-cm}}$.
\end{proposition}

 We conjecture that Problem~\ref{FTFZ} has a negative answer. Thus, in order to justify fracterm flattening in ring-like structures of the form 
 $\mathsf{Enl}_\bot(R_{/})$ we find an incentive to contemplate eager equality. Although eager equality has a problematic logic -- because it is not transitive --  nevertheless we consider eager equality to provide a meaningful semantics for fracterm calculus for this general division over an integral domain. In other words: for plausible structures with division 
 -- in particular, for $\mathsf{Enl}_\bot(R_{/})$ with $R$ an integral domain -- the (very plausible!) equational logic of $E_{\mathsf{ftc-cm}}$ is not sound w.r.t. to Tarski semantics, while it is sound w.r.t. the (less plausible!) eager equality semantics. 
 
Actually, these observations create a good conceptual case for studying and deploying eager equality. They match with the observation in~\cite{BergstraT2023_CJ} that working with rewrite rules which are sound for eager equality provides a plausible approach to term rewriting for common meadows.

\subsection{Eager equality and conditional equations}

Notice that $\mathsf{AVL}$ fails under eager equality in a common meadow of rational numbers, say on $x=2$. 
More technically, consider the conditional equation $\phi \equiv 0 = \bot \to 0 = 1$. We notice that $\phi$ is valid in all common meadows because the premise never holds. At the same time, however, using eager equality instead of ordinary equality, $\phi$ is invalid in all common meadows because the premise always holds while the conclusion never holds. We further notice that the conditional equation  $\phi_{2,3,5}$ as introduced in section \ref{completeness_conjecture} is valid in all common meadows under eager equality.



\section{Concluding remarks}\label{concluding_remarks}

We will reflect on the diversity of algebraic models and logical theories suitable for computer arithmetics.


\subsection{Various forms of total divisions}

Given the original motivation to study arithmetical structures as abstract data types suitable for computer specification and computation, the classical algebra of fields is an obvious starting place: field theory provides a comprehensive axiomatic theory of the rational, real and complex numbers, and calculating with polynomials and solving polynomial equations. This theory is widely taken to be the foundation for what many creators and users of mathematics rely on in their work.\footnote{Among the commonly taught specialsations of field theory we note: Galois theory, algebraic number theory, algebraic geometry, model theory of field axioms.} However, in computing, as we noted in the Introduction, it is in need of modification to be `fit for purpose' as a tool for the specification of data types.

The modification needed and studied here is to add division and make it total.  Elsewhere we have studied several ways of accomplishing this for fields by making different semantic choices for $\frac{x}{0}$. Using various semantical values to be found in practical computations  -- such as $\mathsf{error}$, $\infty$, $\mathit{NaN}$, the last standing for `not a number' --  we have constructed equational specifications (under initial algebra semantics) for the following data types of meadows of rational numbers:

\textit{Involutive meadows}, where an element of the meadow's domain is used for totalisation, in particular $1/0 = 0$, \cite{BergstraT2007}.

\textit{Common meadows}, the subject of this paper, where a new external element $\bot$ that is `absorptive' is used for totalisation $1/0 = \bot$, \cite{BergstraP2015}. 

\textit{Wheels}, where one external $\infty$ is used for totalisation $1/0 = \infty = -1/0$, together with an additional external error element  $\bot$ to help control the side effects of infinity, \cite{Setzer1997,Carlstroem2004,BergstraT2021a}. 

\textit{Transrationals}, where besides the error element $\bot$, two external signed infinities are added, one positive and one negative, so that
division is totalised by setting $1/0 = \infty$ and $-1/0 = -\infty$, \cite{AndersonVA2007,dosReisGA2016,BergstraT2020}. 

In practice, these models are based on data type conventions to be found in theorem  provers, common calculators, exact numerical computation and, of course, floating point computation, respectively. For some historical remarks on division by zero, we mention~\cite{AndersonB2021}, and for a survey we mention~\cite{Bergstra2019b}.

%
%

Our results here and elsewhere point to the fact that arithmetical abstract data types with error values are theoretically superior among the many practical conventions we have studied. The other meadows run into fundamental algebraic difficulties that the common meadows do not. For example, a problem for the involutive case is the failure of flattening composed fractional expressions -- compare section \ref{fracterms}. In~\cite{BergstraBP2013} it is shown that, with the axioms for involutive meadows, terms are provably equal to  \textit{only finite sums} of flat fracterms; and  in~\cite{BergstraM2016a}, it is shown that \textit{arbitrarily large numbers of summands of flat fracterms} are needed for that purpose.

We thank Jo\~{a}o Dias, Bruno Dinis and Alban Ponse for several useful comments made by them on two preparatory drafts of the paper.

\end{document}